\documentclass[a4paper,USenglish]{article} %

\usepackage{amstext}
\usepackage{amsthm}
\usepackage{amsmath}
\usepackage{amssymb}
\usepackage{todonotes}
\usepackage{hyperref}

\usepackage{algorithm} %
\usepackage[noend]{algpseudocode} %
\usepackage{authblk}

\newcommand{\var}{\mathrm{var}}
\newcommand{\sol}{\operatorname{sol}}
\newcommand{\triv}{\mathrm{triv}}

\newcommand{\form}{\varphi} %
\newcommand{\Fform}{F_\varphi}

\newcommand{\coclone}[1]{\langle #1 \rangle}
\makeatletter
\usepackage{tikz}
\usepackage{tikz-qtree}
\usetikzlibrary{arrows.meta}

\usepackage{apxproof}

\newtheoremrep{theorem}{Theorem}
\newtheoremrep{lemma}[theorem]{Lemma}
\newtheoremrep{corollary}[theorem]{Corollary}

\newtheorem{proposition}[theorem]{Proposition}
\newtheorem*{claim*}{Claim}
\newtheorem*{fact*}{Fact}
\newtheorem{claim}[theorem]{Claim}

\theoremstyle{definition}
\newtheorem{definition}[theorem]{Definition}
\newtheorem{example}[theorem]{Example}
\AtBeginEnvironment{example}{%
    \pushQED{\qed}%
}
\AtEndEnvironment{example}{\popQED\endexample}

\newtheorem{observation}[theorem]{Observation}

\makeatother

\title{A characterization of efficiently compilable constraint languages}
\author[1]{Christoph Berkholz\thanks{Funded by the Deutsche Forschungsgemeinschaft (DFG, German Research Foundation) - project number 414325841.}}
\author[2]{Stefan Mengel\thanks{Partially supported by the ANR project EQUUS ANR-19-CE48-0019.}} 
\author[1]{Hermann Wilhelm\thanks{Funded by the Deutsche Forschungsgemeinschaft (DFG, German Research Foundation) - project number 414325841.}}
\affil[1]{Technische Universit{\"a}t Ilmenau, Ilmenau, Germany}
\affil[2]{Univ.~Artois, CNRS, Centre de Recherche en Informatique de Lens (CRIL)}

\begin{document}
\maketitle

\begin{abstract}
A central task in \emph{knowledge compilation} is to compile a
CNF-SAT instance into a succinct representation
format that allows efficient operations such as testing satisfiability, counting, or enumerating all
solutions.
Useful representation formats studied in this area range from ordered
binary decision diagrams (OBDDs) to circuits in decomposable negation normal
form (DNNFs).

While it is known that there exist CNF formulas that require
exponential size representations, the situation is less well studied
for other types of constraints than Boolean disjunctive clauses.
The \emph{constraint satisfaction problem} (CSP) is a powerful framework that generalizes
CNF-SAT by allowing arbitrary sets of constraints over any finite
domain.
The main goal of our work is to understand for which type of
constraints (also called the \emph{constraint language}) it is possible to
efficiently compute representations of polynomial size.
We answer this question completely and prove two tight characterizations of
efficiently compilable constraint languages, depending on whether
target format is structured.

We first identify the combinatorial property of ``strong blockwise
  decomposability'' and show that if a constraint language has this
property, we can compute DNNF representations of linear size. For all other
constraint languages we construct families of CSP-instances that provably
require DNNFs of exponential
size. For a subclass of ``strong \emph{uniformly} blockwise
  decomposable'' constraint languages we obtain a similar dichotomy for
  \emph{structured} DNNFs.
In fact, strong (uniform) blockwise decomposability even allows efficient
compilation into multi-valued analogs of OBDDs and FBDDs,
respectively. Thus, we get complete characterizations for all knowledge
compilation classes between O(B)DDs and DNNFs.
\end{abstract}

\newpage

\tableofcontents

\section{Introduction}

One of the main aims of \emph{knowledge compilation} is to encode
solution sets of computational problems into a succinct but usable
form~\cite{DarwicheM02}. Typical target formats for this compilation
process are different forms of decision
diagrams~\cite{Wegener00} or restricted classes of Boolean
circuits. One of the most general representation formats are circuits in
\emph{decomposable negation normal form} (DNNF), which have been
introduced in \cite{Darwiche01} as a compilation target for Boolean functions. 
Related notions, which also rely on the central decomposability
property, have been independently considered in databases~\cite{OlteanuZ15,Olteanu16} and
circuit complexity~\cite{RazSY08,AlonKV20}. Besides these, DNNF circuits and related
compilation classes have also been proven useful in areas like probabilistic inference~\cite{BroeckS17}, constraint satisfaction~\cite{KoricheLMT15, BerkholzV23, AmilhastreFNP14, MateescuD06}, MSO evaluation~\cite{AmarilliBJM17}, QBF
solving~\cite{CapelliM19}, to
name a few. 
The DNNF representation format has become a popular
data structure because it has a particularly good balance
between generality and usefulness~\cite{DarwicheM02}.
There has also been a large amount of practical work on compiling solution
sets into DNNF or its fragments, see
e.g.~\cite{LagniezM17,Darwiche04,MuiseMBH12,ChoiD13,OztokD15,KoricheLMT15}. In
all these works, it is assumed that the solutions to be compiled are
specified by a system of constraints, often as a set of disjunctive
Boolean clauses, i.\,e., in conjunctive normal form (CNF).

In this setting, however, strong lower bounds are known: it was shown that
there are CNF-formulas whose representation as DNNF requires
exponential
size~\cite{BovaCMS14,BovaCMS16,Capelli17,AmarilliCMS20}. The
constraints out of which the hard instances
in~\cite{BovaCMS14,AmarilliCMS20,Capelli17} are constructed are very
easy---they are only monotone clauses of two variables. Faced with
these negative results, it is natural to ask if there are any classes
of constraints that guarantee efficient compilability into DNNF.

We answer this question completely and prove a tight characterization for every constraint language $\Gamma$.
We first examine the combinatorial property of \emph{strong blockwise decomposability} and show that if a constraint language $\Gamma$ has this
property, any system of constraints over $\Gamma$ can be compiled into
a DNNF representation of linear size within polynomial time.
Otherwise, there are systems of
constraints that require exponential size DNNF representations. In the
tractable case, one can even compile to the restricted fragment of
\emph{free decision diagrams} (FDD) that are in general known to be
exponentially less succinct than DNNF~\cite{Wegener00,DarwicheM02}.

Additionally, we also consider the important special case of so-called structured
DNNF~\cite{PipatsrisawatD08} which are a generalization of the
well-known ordered binary decision diagrams (OBDD)~\cite{Bryant86}. We show
that there is a restriction of strong blockwise decomposability that
determines if systems of constraints over a set $\Gamma$ can be
compiled into structured DNNF in polynomial time. In the tractable case, we can in fact
again compile into a restricted fragment, this time ordered decision
diagrams (ODD).

Furthermore, we separate both notions of constraint languages admitting
structured and only unstructured representations and thus give a complexity picture of
the tractability landscape. We also show that it is decidable, whether
a given constraint language strongly (uniformly) blockwise
decomposable (a question left open in the conference version \cite{DBLP:conf/stacs/BerkholzMW24}).
Let us stress that all our lower bounds provide unconditional size lower
bounds on (structured) DNNFs and thus do not depend on unproven complexity assumptions.

\subparagraph*{Further related work.}  Our work is part of a long line
of work in constraint satisfaction, where the goal is to precisely
characterize those constraint languages that are
``tractable''. Starting with the groundbreaking work of
Schaefer~\cite{Schaefer78} which showed a dichotomy for deciding
consistency of systems of \emph{Boolean} constraints, there has been much
work culminating in the dichotomy for general constraint
languages~\cite{Bulatov17,Zhuk17}. Beyond decision, there are
dichotomies for counting~\cite{CreignouH96,Bulatov13,DyerR13},
enumeration~\cite{CreignouOS11},
optimization~\cite{Creignou95,KhannaSTW00,CreignouKS2001} and
in the context of valued CSPs \cite{ThapperZ16,DBLP:journals/jacm/CaiC17}. Note that the hardness part (showing that certain
constraint languages do not admit efficient algorithms) always relies
on some complexity theoretic assumption.

The complexity of constraint satisfaction has also been studied ``from the
other side'', where the constraint language is unrestricted and the
structure of the \emph{constraint network} (i.\,e. how the constraints
are arranged) has been analyzed. In this setting, characterizations of
(bounded arity) constraint networks have been obtained for the
deciding the existence of solutions \cite{Grohe2007} and counting solutions \cite{DalmauJ04}
(again, under some complexity theoretic assumption), while for 
enumerating solutions only partial results exist
\cite{DBLP:journals/jcss/BulatovDGM12}. Tractability classifications of the constraint network have also been
obtained in the context of valued CSPs \cite{CarbonnelRZ22}. 
Recently, an unconditional characterizations of constraint networks
that allow efficient compilation into DNNFs has been proven
\cite{BerkholzV23,DBLP:conf/icdt/BerkholzV26}.

Out of these works, our results are most closely related
to the counting dichotomy of~\cite{DyerR13}. The tractable classes we obtain in our dichotomies are also
tractable for counting, and in fact we use the
known counting algorithm in our compilation algorithms
as a subroutine. Also, the tractability criterion has a similar flavor, but
wherever in~\cite{DyerR13} the count of elements in certain relations
is important, for our setting one actually has to understand their
structure and how exactly they decompose. That said, while the
tractability criterion is related to the one in~\cite{DyerR13}, the
techniques to show our results are very
different. In particular, where~\cite{DyerR13} use reductions from
\#P-complete problems to show hardness, we
make an explicit construction and then use communication complexity to
get strong, unconditional lower bounds. Also, our two algorithms for
construction ODDs and FDDs are quite different to the counting algorithm.

\subparagraph*{Outline of the paper.}  After some preliminaries in
Section~\ref{sct:preliminaries}, we introduce the decomposability
notions that we need to formulate our results in
Section~\ref{sct:decomposability}. We also give the formal formulation
of our main results there. In Section~\ref{sct:properties}, we show
some of the properties of the constraints we defined in the section
before. These properties will be useful throughout the rest of the
paper. In Section~\ref{sct:algorithms}, we present the algorithms for
the positive cases of our dichotomies, then, in
Section~\ref{sct:lower}, we show the corresponding lower bounds. We
specialize our results to the case of Boolean relations in
Section~\ref{sct:boolean}, showing that the tractable cases are
essentially only equalities and disequalities. The decidability of our dichotomy criterion is treated in Section~\ref{sct:decidability}. Finally, we conclude in
Section~\ref{sct:conclusion}.

\section{Preliminaries}\label{sct:preliminaries}

\subparagraph*{Constraints.} Throughout the paper we let $D$ be a finite \emph{domain} and $X$ a
finite set of \emph{variables} that can take values over
$D$. We typically denote variables by $u,v,w,x,y,z$ and
domain elements by $a,b,c,d$. A $k$-tuple $(x_1,\ldots,x_k)\in X^k$ of
variables is also denoted by
$\vec{x} = x_1x_2\cdots x_k$ and we let
$\tilde{x}:=\{x_1,\ldots,x_k\}$ be the set of variables occurring in $\vec{x}$. We use the same notation for tuples of domain elements. 
A \emph{$k$-ary relation} $R$ (over $D$) is a set $R\subseteq
D^k$. A \emph{$k$-ary constraint} (over $X$ and $D$), denoted by $R(\vec{x})$, consists
of a $k$-tuple of (not necessarily distinct) variables
$\vec{x}\in X^k$ and a $k$-ary relation $R\subseteq D^k$.
For a constraint $R(\vec{x})$ we call $\tilde{x}$ the \emph{scope} of
the constraint and $R$ the \emph{constraint relation}.
The \emph{solution
set} of a constraint $R(x_1,\ldots,x_k)$ is defined by
$\operatorname{sol}(R(x_1,\ldots,x_k)) := \{\beta \mid \beta\colon
\{x_1,\ldots,x_k\}\to D;\, (\beta(x_1),\ldots,\beta(x_k))\in R\}$.
Since the order of the columns in a constraint relation is not of great
importance for us, we often identify constraints with their solution
set and
treat sets $\mathcal S$ of mappings from $Y\subseteq X$ to $D$ as a constraint with
scope $Y$ and solution set $\mathcal S$. Moreover, for readability we will often write,
e.\,g., 
``a constraint $R(x,y,\vec{w})$, such that $x$ and $y$ \ldots'' when
we do not strictly require $x$ and $y$ to be at the first and second
position and actually refer to any constraint having $x$ and $y$ in its scope.

\subparagraph*{CSP-instance and constraint language.} A \emph{constraint satisfaction instance} $I=(X,D,C)$ consists of a
finite set of variables $X$, a finite domain $D$ and a finite set $C$
of constraints. The solution set $\operatorname{sol}(I):= \{\alpha
\mid \alpha\colon X\to D;\; \alpha|_{\tilde{x}} \in
\operatorname{sol}(R(\vec{x}))\text{ for all } R(\vec{x})\in C\}$ of an
instance $I$ is the set of mappings from $X$ to $D$ that satisfies all
constraints.
A \emph{constraint language} $\Gamma$ is a finite set of relations (over some
finite domain $D$). A constraint satisfaction instance $I$ is
a $\operatorname{CSP}(\Gamma)$-instance if every constraint relation
$R$ occurring in $I$ is contained in the constraint language $\Gamma$.

\subparagraph*{Conjunction and projection.} A \emph{conjunction} $R(\vec{u}) = S(\vec{v})\land
T(\vec{w})$ of two constraints $S(\vec{v})$ and $T(\vec{w})$ defines a
constraint $R(\vec{u})$ over $\vec{u}=u_1\cdots u_\ell$ with scope
$\tilde{u}=\tilde{v} \cup \tilde{w}$ and
constraint relation $R := \{(\beta(u_1),\ldots,\beta(u_\ell)) \mid
\beta\colon \tilde{u}\to D;\; \beta|_{\tilde v}\in
\operatorname{sol}(S(\vec{v})) \text{ and } \beta|_{\tilde w}\in
\operatorname{sol}(T(\vec{w}))\}$.\footnote{Again, we may just write ``$S(\vec{v})\land
T(\vec{w})$'' instead of ``$R(\vec{u}) = S(\vec{v})\land
T(\vec{w})$'' if the order or multiple occurrences of variables in
$\vec{u}$ does not matter (the solution set $\operatorname{sol}(R(\vec{u}))$ is
always the same).}
If $\tilde{v}\cap\tilde{w}=\emptyset$, then a conjunction is called a
\emph{(Cartesian) product} and written $R(\vec{u}) = S(\vec{v})\times
        T(\vec{w})$. 
For a constraint $R(\vec{x})$ and a set $Y\subseteq \tilde x$ we let
the \emph{projection} $\pi_Y R(\vec{x})$ be the constraint $S(\vec{y})$ with scope
$\tilde y = Y$ obtained by projecting the constraint relation to corresponding
coordinates, i.\,e., for $\vec{x}=x_1\cdots x_k$ let $1\leq i_1<\cdots
< i_\ell\leq k$ s.\,t. $\{i_1,\ldots,i_\ell\} = \{i\mid
x_i\in Y\}$, $\vec{y}:=x_{i_1}\cdots x_{i_\ell}$, and $S :=
\{(a_{i_1},\ldots,a_{i_\ell})\mid (a_1,\ldots,a_k)\in R\}$.

\subparagraph*{Formulas and pp-definability.}
Sometimes it is necessary to treat conjunctions and projections of
constraints as purely syntactical objects. To make this formal,
we associate to every $k$-ary relation $R$ a \emph{relation symbol}
$\dot R$ and denote by $\dot R(x_1,\ldots,x_k)$ the \emph{relational
  atom} that is interpreted by the constraint $R(x_1,\ldots,x_k)$. 

A \emph{$\Gamma$-formula} $\form(\vec x)=\bigwedge_{i} \dot S_i(\vec{x}_i)$ is a conjunction over several
relational atoms whose relations $S_i$ are in $\Gamma$. The $\Gamma$-formula
$\form(\vec x)$ defines the constraint $\Fform(\vec x)=\bigwedge_{i} S_i(\vec{x}_i)$.
Note that any
CSP($\Gamma$)-instance $I=(X,D,C)$ corresponds to a $\Gamma$-formula
$\varphi_I(\vec{x}) = \bigwedge_{R_i(\vec{x}_i)\in C}\dot R_i(\vec{x}_i)$ with
$\tilde x = X$ and
$\operatorname{sol}(F_{\varphi_I}(\vec{x}))=\operatorname{sol}(I)$.
Thus, we can treat CSP($\Gamma$)-instances as $\Gamma$-formulas and
vice versa.

A \emph{primitive positive (pp)} formula over $\Gamma$ is an
expression of the form $\form(\vec{y}) = \pi_{\tilde y}(\bigwedge_{i}\dot S_i(\vec{x}_i)
\land \bigwedge_{(j,k)}x_j=x_k)$ consisting of a projection applied to
a $\Gamma\cup \big\{\{(a,a)\mid a\in D\}\big\}$-formula, which uses
constraint relations from $\Gamma$ and the equality constraint.
Again, the formula $\form(\vec{y})$ naturally defines the constraint
$\Fform(\vec{y})$. Note that conjunctions of pp-formulas can be written
as pp-formula by putting one projection at the front and renaming
variables. A constraint is \emph{pp-definable} over $\Gamma$, if it
can be defined by a pp-formula over $\Gamma$. 
Moreover, a relation $R\subseteq D^k$ is \emph{pp-definable} over
$\Gamma$ if it is the constraint relation of a pp-definable constraint
$R(x_1,\ldots,x_k)$ for pairwise distinct
$x_1,\ldots,x_k$. The co-clone $\langle\Gamma\rangle$ of $\Gamma$ is the set
of all pp-definable relations over $\Gamma$.

Since it is clear from the context whether we deal with formulas (that
we always denote with greek letters $\varphi,\psi,\theta$) instead of
constraints (denoted by uppercase letters), we will omit the dot
 to formally distinguish atoms $\dot R(\vec x)$ and
constraints $R(\vec x)$ from now on.

\subparagraph*{Selection.} It is often helpful to use additional unary relations $S\subseteq D$
and to write
$U_S(x)$ for the constraint $S(x)$ with scope $\{x\}$ and constraint
relation $S$. We use $U_a(x)$ as an abbreviation for $U_{\{a\}}(x)$.
For a constraint $R(\vec{u})$, a variable $x\in \tilde u$, and a
domain element $a\in D$ we define
the \emph{selection} $R(\vec{u})|_{x=a}$ be the constraint obtained by
forcing $x$ to take value $a$, that is, $R(\vec{u})|_{x=a} :=
R(\vec{u}) \land U_{a}(x)$. Similarly, for a set $S\subseteq D$ we
write $R(\vec{u})|_{x\in S} := R(\vec{u}) \land U_{S}(x)$.

\subparagraph*{DNNF.}
We will be interested in circuits representing assignments of variables to a finite set of values. To this end, we introduce a multi-valued variant of DNNF; we remark that usually DNNFs are only defined over the Boolean domain $\{0,1\}$~\cite{Darwiche01}, but the extension we make here is straightforward and restricted variants have been studied e.g.~in~\cite{KoricheLMT15,AmilhastreFNP14,MateescuDM08,FargierM06} under different names.

Let $X$ be a set of variables and let $D$ be a finite set of values. A circuit $O$ over the operations $\times$ and $\cup$ is a directed acyclic graph with a single sink, called output gate, and whose inner nodes, called gates, are labeled with $\times$ or $\cup$. We assume that gates labeled with $\times$ have exactly two incoming edges. For gates labeled with $\cup$ we only assume that they have at least one incoming edge. The source-nodes, called inputs of the circuit, are labeled by expressions of the form $x\mapsto a$ where $x\in X$ and $a\in D$. A \emph{sub-circuit} $O'$ is a circuit that we get by choosing a gate $v$ from $O$ as a new sink and letting $O'$ contain all gates of $O$ from which there is a path to $v$. In this construction, the gates that are both in $O$ and $O'$ must have the same labels. We say that $\times$-gate $v$ is \emph{decomposable} if there are no two inputs labeled with $x\mapsto a$ and $x\mapsto b$ (with possibly $a=b$) that have a path to $v$ going through different children of $v$. A DNNF is a circuit in which all $\times$-gates are decomposable. Note that in a DNNF, for every $\times$-gate $v$, every variable $x\in X$ can only appear below one child of $v$. Clearly, every sub-circuits of a DNNF is also a DNNF.

For every DNNF $O$ we define the set $S(O)$ of assignments \emph{captured} by $O$ inductively:
\begin{itemize}
        \item The set captured by an input $v$ with label $x\mapsto d$ is the single element set $S(v)=\{x\mapsto d\}$.
        \item For a $\cup$-gate with children $v_1, \ldots , v_\ell$,  we set $S(v):= S(v_1)\cup \ldots \cup S(v_\ell)$.
        \item For a $\times$-gate with children $v_1, v_2$,  we set $S(v):= S(v_1)\times S(v_2)$ where for two assignments $\alpha: X_1\rightarrow D$ and $\beta:X_2\rightarrow D$, we interpret $(\alpha,\beta)$ as the joined assignment $\gamma:X_1\cup X_2\rightarrow D$ with 
     $           \gamma(x):= \begin{cases} \alpha(x), & \text{ if } x\in X_1\\ \beta(x), & \text{ if } x\in X_2\end{cases}.$
\end{itemize}
We define $S(O):= S(v_o)$ where $v_o$ is the output gate of $O$.
Note that $S(O)$ is well-defined in the case of $\times$, since $X_1$ and $X_2$ are disjoint because of decomposability. 
Finally, we say that $O$ accepts an assignment $\alpha$ to $X$ if
there is an assignment $\beta\in S(O)$ such that $\alpha$ is an
extension of $\beta$ to all variables in $X$.
We say that $O$ \emph{represents} a constraint $C(\vec x)$ if $O$ accepts exactly the assignments in $\sol(C(\vec x))$. 

A v-tree of a variable set $X$ is a rooted binary tree whose leaves are in bijection to $X$. For a v-tree $T$ of $X$ and a node $t$ of $T$ we define $\var(t)$ to be the variables in $X$ that appear as labels in the subtree of $T$ rooted in $t$. We say that a DNNF $O$ over $X$ is \emph{structured by $T$} if for every sub-circuit $O'$ of $O$ there is a node $t$ in $T$ such that $O'$ is exactly over the variables $\var(t)$~\cite{PipatsrisawatD08}. We say that $O$ is \emph{structured} if there is a v-tree $T$ of $X$ such that $O$ is structured by $T$.

In the constructions in the remainder of this section, it will be convenient to have a slightly generalized variant of DNNF: besides the inputs of the form $\{x\mapsto a\}$, we also allow inputs $\emptyset$ and $\{\varepsilon\}$. To extend the semantics of DNNF to these new inputs, we define $S(\emptyset) = \emptyset$ and $S(\varepsilon):= \{\varepsilon\}$ where $\varepsilon$ is the unique assignment to $D$ with empty variable scope. We observe that this slight extension does not increase the expressivity of the model.
    
\begin{lemma}\label{lem:emptysetgate}
    Every DNNF with inputs $\emptyset$ and $\{\varepsilon\}$ can be turned into a DNNF without such inputs in polynomial time.
\end{lemma}
\begin{proof}
    We first show this for $\emptyset$-gates. So let $O$ be a DNNF and $v$ an $\emptyset$-input in $O$. There are two cases: if the parent of $v$ is an $\cup$-gate $u$, then we delete $v$. If there is still another child of $u$, we are done in this case. Otherwise, we change the label from $\cup$ to $\emptyset$ and continue iteratively. If $u$ is a $\times$-gate, then by definition $S(u) = \emptyset$. In this case, we delete $v$, we delete the second incoming edge to $u$ and relabel $u$ to be an $\emptyset$-gate. By definition of the semantics, none of the cases changes the assignments captured in any of the gates still in the DNNF after the operation. 
    
    To eliminate the $\{\varepsilon\}$-inputs in a DNNF $O$ we apply the following rules in any order while possible:
    \begin{enumerate}
        \item \label{rule1}If there is a $\cup$-gate $v$ with an input $v'$ with label $\{\varepsilon\}$, we substitute $v$ by $v'$ in all its parents.
        \item \label{rule2} If there is a $\times$-gate with children $v_1, v_2$ such that $v_2$ is labeled with $\{\varepsilon\}$, we substitute $v$ by $v_1$ in all parents of $v$.
    \end{enumerate}
    
    Let $O'$ a circuit we get when we cannot apply any of the rules anymore.
    Clearly applying either of the rules does not break the decomposability condition, so $O'$ is still a DNNF. We claim that it also accepts the same assignments. To see this, consider the DNNF $O''$ we get after applying a single step of the rules. For Rule~\ref{rule2}, it is immediate that the assignments accepted by $O''$ are the same as for $O$, since for every assignment $\alpha$ we have that $(\alpha, \varepsilon)$ is identical to $\alpha$ and thus $S(v)= S(v_1)$, so $S(O) = S(O'')$. For Rule~\ref{rule1} the situation is slightly less clear since $S(O)$ and $S(O'')$ might differ. Since $S(O'')\subseteq S(O)$ by construction, we have that every assignment accepted by $O''$ is also accepted by~$O$. So consider the other direction, i.e., an assignment $\alpha$ accepted by $O$. Then~$\alpha$ is an extension of some assignment $\alpha'\in S(v_o)$ in $O$, so captured by the root of $O$. If $\alpha'$ is also captured by the root of $O''$, we are done, so let us assume that $\alpha'$ is not captured there. Then $\alpha'$ must decompose as $(\beta,\gamma)$ where $\gamma\in S(v)$ since this is the only place we changed in Rule~\ref{rule1}. But then $\alpha$ is also an extension of $(\beta, \varepsilon) = \beta$ which is captured by $v_o$ both in $O$ and $O''$, so $\alpha$ still gets accepted in~$O''$.
    So we have shown that $O''$ accepts the same assignments of $O$ and thus by induction the same is true for $O'$.
    
    To finish the proof, assume that in $O'$ there are no gates that are not connected to the output gate $v_o$ (deleting such gates obviously does not change the accepted assignments). Then the only case in which any gates with label $\{\varepsilon\}$ can occur if this is the case of $v_o$ itself. But in this trivial case, $O$ and thus $O'$ accept all assignments so we can substitute $O$ by an empty circuit. In all other cases, we have eliminated all $\{\varepsilon\}$-gates, as desired.
\end{proof}
We will use the following basic results on DNNF which correspond to projection and selection on constraints.

\begin{lemma}\label{lem:restrictDNNF}
    Let $C(\vec u)$ be a constraint for which there exists a \textup{DNNF} of size $s$, $x \in \tilde u$ and $|\tilde u| \ge 2$ and let $A \subseteq D$. Then there exists a \textup{DNNF} for the selection $C(\vec u)|_{x\in A}$ of size at most $s$. 
\end{lemma}
\begin{proof}
Let $O$ be the DNNF computing $C(\vec u)$. We get a new DNNF $O'$ by simply substituting all input labels in~$O$ that have the form $x\mapsto a$ for $a\notin A$ by $\emptyset$. A simple induction shows that for every gate $v$ of $O$ we then have that $S(v)$ contains exactly all assignments accepted by $v$ in $O$ that do not assign a value from $D\setminus A$ to $x$. But since $O$ accepts $C(\vec u)$ and thus all accepted assignments assign a value to $x$, $O'$ accepts exactly the assignments in $C(\vec u)$ that assign values in $A$ to $x$. The claim directly follows with Lemma~\ref{lem:emptysetgate}.
\end{proof}

\begin{lemma}\label{lem:projectDNNF}
    Let $C(\vec u)$ be a constraint for which there exists a \textup{DNNF} of size $s$ and $\tilde v \subseteq \tilde u$. Then there exists a \textup{DNNF} for the projection $\pi_{\tilde v}(C(\vec u))$ of size at most $s$.
\end{lemma}

For the proof of Lemma~\ref{lem:projectDNNF}, we will use the concept of \emph{proof trees}, which are a classical notion from circuit complexity, sometimes under different names like \emph{parse trees} or simply \emph{certificates}, see e.g.~\cite{Venkateswaran87,MalodP08,BovaCMS16}. A proof tree $T$ of a DNNF $O$ is defined to be a subcircuit of $O$ that
\begin{itemize}
    \item contains the output gate of $O$,
    \item contains, for every $\times$-gate of $O$, both children of that gate,
    \item contains, for every $\cup$-gate of $O$, exactly one child of that gate, and
    \item contains no other gate.
\end{itemize}
Note that $O$ may in general have an exponential number of proof trees. Every proof tree of $O$ is a DNNF that has as its underlying graph a tree. As a DNNF,~$T$ captures a set of assignments to the variables of $O$, and, from the definition of proof trees, we directly get $S(T)\subseteq S(O)$. Moreover, if we let $\mathcal T(O)$ denote all proof trees of $O$, then 
\begin{align}
    S(O):= \bigcup_{T\in \mathcal T(O)} S(T),\label{eq:prooftrees}
\end{align}
so every assignment captured by $O$ gets captured by (at least) one of the proof trees of $O$.

\begin{proof}
Assume w.l.o.g.~that $\tilde u \setminus \tilde v$ contains only a single variable $x$. If this is not the case, we can simply iterate the construction, projecting away one variable at a time.    

Let $O$ be the DNNF computing $C(\vec u)$. We construct a new DNNF $O'$ over $\tilde v$ by substituting every input $x\mapsto a$ for $a\in D$ by $\varepsilon$. We claim that $O'$ computes $\pi_{\tilde v}(C(\vec u))$. Since up to the input labels $O$ and $O'$ have the same proof trees, it suffices to show the claim for every proof tree. So let $T$ be a proof tree in $O$ and $T'$ the corresponding proof tree in $O'$. 
If $T$ does not contain an input of the form $x\mapsto a$, then $T$
and $T'$ are equal, so there is nothing to show. If $T$ contains such
an input, then it contains exactly one of them due to
decomposability. Let $\alpha$ be an assignment captured by $T$, then
$\alpha|_{\tilde v}$ is captured by $T'$, so $\pi_{\tilde
  v}(S(T))\subseteq S(T')$. For the other direction, fix $\beta\in
S(T')$. Then $\beta$ does not assign a value to $x$. But by assumption
there is an input $x\mapsto a$ in $T$, so all assignments in $S(T)$
assign $x$ to $a$. In particular, there is an assignment $\alpha\in
S(T)$ that extends $\beta$ to $x$. Thus $\beta\in \pi_{\tilde
  v}(S(T))$ and the claim follows.
\end{proof}

\subparagraph*{Decision Diagrams.}

A decision diagram $O$ over a variable set $X$ is a directed acyclic graph with a single source and two sinks and in which all non-sinks have $|D|$ outgoing edges. The sinks are labeled with $0$ and $1$, respectively, while all other nodes are labeled with variables from $x$. For every non-sink, the $|D|$ outgoing edges are labeled in such a way that every value in $D$ appears in exactly one label. Given an assignment $\alpha$ to $X$, the value computed by $O$ is defined as follows: we start in the source and iteratively follow the edge labeled by $\alpha(x)$ where $x$ is the label of the current node. We continue this process until we end up in a leaf whose label then gives the value of $O$ on $\alpha$. Clearly, in this way~$O$ computes a constraint over $X$ with relation $\{\alpha \mid O \text{ computes } 1 \text{ on input } \alpha\}$. 

We are interested in decision diagrams in which on every source-sink-path every variable appears at most once as a label. We call these diagrams \emph{free decision diagrams (FDD)}. An FDD for which there is an order $\pi$ such that when a variable $x$ appears before $y$ on a path then $x$ also appears before $y$ in $\pi$ is called an \emph{ordered decision diagram}. We remark that FDD and ODD are in the literature mostly studied for the domain $\{0,1\}$ in which case they are called FBDD and OBDD, respectively, where the ``B'' stands for binary. Note also that there is an easy linear time translation of FDD into DNNF and ODD into structured DNNF, see e.g.~\cite{DarwicheM02}. In the other direction there are no efficient translations, see again~\cite{DarwicheM02}.

\section{Blockwise decomposability}\label{sct:decomposability}

In this section we introduce our central notions of blockwise and
uniformly blockwise decomposable constraints and formulate our main theorems that lead to a characterization of
efficiently representable constraint languages.

The first simple insight is the following. Suppose two constraints
$S(\vec{v})$ and $T(\vec{w})$ with disjoint scopes are efficiently
representable, e.\,g., by a small ODD. Then their Cartesian product
$R(\vec{v},\vec{w})=S(\vec{v}) \times T(\vec{w})$ also has a small ODD: given an assignment $(\alpha,\beta)$, we just need
to check independently whether $\alpha\in S$ and $\beta\in T$, for
example, by first using the ODD for $S(\vec{v})$ and then using the
ODD for $T(\vec{w})$. Thus, if a constraint can be expressed as a
Cartesian product of two constraints, we only have to investigate
whether the two parts are easy to represent. This brings us to our
first definition. 

\begin{definition}
Let $R(\vec{u})$ be a constraint and $(V_1,\ldots,V_\ell)$ be a partition of its scope.
We call $R(\vec{u})$ \emph{decomposable w.r.t.~$(V_1,\ldots,V_\ell)$} if 
$R(\vec{u}) = \pi_{V_1}(R(\vec{u})) \times \cdots \times
\pi_{V_\ell}(R(\vec{u}))$.
A constraint $R(\vec{u})$ is \emph{indecomposable} if it is only
decomposable w.r.t.~trivial partitions $(V_1,\ldots,V_\ell)$ where
$V_i=\emptyset$ or $V_i=\tilde{u}$ for $i\in[\ell]$.
\end{definition}

Next, we want to relax this notion to constraints that are ``almost''
decomposable. Suppose we have four relations $S_1,S_2$ of arity $s$
and $T_1,T_2$ of arity $t$ and let $a,b$ be two distinct domain
elements. Let
\begin{equation}
  \label{eq:2}
R := (\{(a,a)\}\times S_1\times T_1) \cup  (\{(b,b)\}\times S_2\times T_2).   
\end{equation}

The constraint $R(x,y,\vec{v},\vec{w})$ may now not be decomposable in any
non-trivial variable partition. However, after fixing values for $x$ and $y$ the
remaining selection $R(x,y,\vec{v},\vec{w})|_{x=c,y=d}$ 
is
decomposable in $(\tilde{v},\tilde{w})$ for any pair $(c,d)\in
D^2$. Thus, an ODD could first read values for $x,y$ and then use ODDs
for $S_1(\vec{v})$ and $T_1(\vec{w})$ if $x=y=a$, ODDs for
$S_2(\vec{v})$ and $T_2(\vec{w})$ if $x=y=b$, or reject
otherwise. This requires, of course, that $S_1(\vec{v})$ and
$S_2(\vec{v})$, as well as $T_1(\vec{w})$ and $T_2(\vec{w})$, have small
ODDs over the \emph{same} variable order. For FDDs and DNNFs, however, we would
not need this requirement on the variable orders.

To reason about the remaining constraints after two variables have been fixed, it is helpful to
use the following matrix notation. Let $R(\vec{u})$ be a
constraint
(of arbitrarily large arity)
and let $x,y\in \tilde{u}$ be two variables in its scope
(we also allow $x=y$).
The \emph{selection
  matrix} $M^{R(\vec{u})}_{x,y}$ is the $|D|\times |D|$ matrix where the rows and
columns are indexed by domain elements $a_i,a_j\in D$ and the entries are
the constraints 
\begin{equation}
  \label{eq:3}
  M^R_{x,y}[a_i,a_j] := \pi_{\tilde{u}\setminus \{x,y\}}(R(\vec{u})|_{x=a_i,y=b_j}).
\end{equation}

\begin{example} \label{example:blockmatrix}
Let $D=\left\{a,b,c\right\}$ and $R(x,y,z,v)$ a constraint with
constraint relation
$
R=\left\{ \left(a,a,a,a\right),\left(b,b,a,b\right),\left(b,b,a,c\right),\left(b,b,c,c\right),\left(c,b,c,a\right)\right\} .
$ %
The selection matrix in $x$ and $y$ is depicted below,
where the first line and column
are the indices from $D$ and the matrix entries contain the
constraint relations of the corresponding constraints $M^{R(x,y,z,v)}_{x,y}[a_i,a_j](z,v)$:
\begin{equation}
  \label{eq:5}
\left(\begin{array}{c|ccc}
x\backslash y & a & b & c\\
\hline a & \left\{ \left(a,a\right)\right\}  & \emptyset & \emptyset\\
b & \emptyset & \left\{ \left(a,b\right),\left(a,c\right),\left(c,c\right)\right\}  & \emptyset\\
c & \emptyset & \left\{ \left(c,a\right)\right\}  & \emptyset
\end{array}\right)
\end{equation}
\end{example}

A \emph{block} in the selection matrix is a subset of rows $A\subseteq
D$ and columns $B\subseteq D$. We also associate with a block $(A,B)$ the
corresponding constraint
$R(\vec{u})|_{x\in A,y\in B}$. 
A selection matrix is a \emph{proper block matrix}, if there exist
pairwise disjoint $A_1,\ldots, A_k \subseteq D$ and pairwise disjoint
$B_1,\ldots, B_k \subseteq D$ such that for all $a_i,a_j\in D$:
\begin{equation}
        \label{eq:6}
        \operatorname{sol}(M^{R(\vec{u})}_{x,y}[a_i,a_j]) \neq \emptyset \quad
        \Longleftrightarrow \quad \text{there is $\ell\in [k]$ such that $a_i\in A_\ell$ and $a_j\in B_\ell$}.
\end{equation}

\begin{example}
The selection matrix in Example~\ref{example:blockmatrix} is a proper
block matrix with $A_1=\{a\}$, $A_2=\{b,c\}$, $B_1=\{a\}$, $B_2=\{b\}$. 
\end{example}

We will make use of the following alternative characterization of
proper block matrices. 
The proof is similar to \cite[Lemma~1]{DyerR13}.

\begin{lemma}\label{lem:2x2blockmatrix}
\label{lem: 2x2 submatrix} A selection matrix $M^{R(\vec{u})}_{x,y}$ is a proper block matrix
if and only if it has no $2\times2$-submatrix with exactly one
empty
entry.
\end{lemma}

\begin{proof}
Let $M^{R(\vec{u})}_{x,y}$ be a proper block matrix and $A_1,\ldots, A_k
\subseteq D$, $B_1,\ldots, B_k \subseteq D$ be the pairwise disjoint
sets provided by the definition. 
A $2\times2$-submatrix $M'$ may intersect $0$, $1$ or $2$ blocks. If it intersects $0$ blocks, then all $4$ of the entries of $M'$ are empty.
If $M'$ intersects only one block, then either $1$, $2$ or $4$ entries are in that block and thus non-empty. So the number of empty entries in $M'$ is $3$, $2$ or $0$.
If $M'$ intersect two blocks, then exactly $2$ of its entries lie in those blocks and are thus non-empty. So $M'$ has $2$ empty entries in that case. 
Overall, $M'$ may have $0$, $2$, $3$, or $4$ empty entries, so it satisfies the claim.

For the other direction, we must show how to find the block
structure. Take one non-empty entry and reorder the rows to $a_1,a_2,\ldots$ and
columns to $b_1,b_2,\ldots$ 
such that this entry is in the first column and first row, and the
first row (column) starts with $s$ ($t$) non-empty entries followed by
empty entries. Consider for all $i,j$ the $2\times2$-submatrix indexed
by $a_1,a_i$ and $b_1,b_j$. Since it does not have exactly
one empty entry, we get:
\begin{itemize}
\item If $i\leq s$ and $j\leq t$, then $M^{R(\vec{u})}_{x,y}[a_i,b_j]\neq\emptyset$.
\item If $i\leq s$ and $j > t$, then $M^{R(\vec{u})}_{x,y}[a_i,b_j]=\emptyset$.
\item If $i> s$ and $j \leq t$, then $M^{R(\vec{u})}_{x,y}[a_i,b_j]=\emptyset$.
\end{itemize}
Thus, we can choose $A_1=\{a_1,\ldots,a_s\}$ $B_1=\{b_1,\ldots,b_t\}$
and proceed with the submatrix on rows $D\setminus A_1$ and columns $D\setminus B_1$ inductively.
\end{proof}

Now we can define our central tractability criterion for constraints
that have small ODDs, namely that any selection matrix is a proper
block matrix whose blocks are decomposable over the same variable partition
that separates $x$ and~$y$.

\begin{definition}[Uniform blockwise decomposability]
        A constraint $R(\vec u)$ is \emph{uniformly blockwise
          decomposable in $x,y$} if $M^{R(\vec u)}_{x,y}$ is a proper block
        matrix with partitions $(A_1,\ldots,A_k)$, $(B_1,\ldots,B_k)$ and there is a partition $(V,W)$ of $\tilde u$ with $x\in V$ and $y\in W$ such that each block $R(\vec{u})|_{x\in A_i,y\in B_i}$ is decomposable in $(V,W)$. A constraint $R(\vec u)$ is \emph{uniformly blockwise decomposable} if it is uniformly decomposable in every pair $x,y \in \tilde u$.
\end{definition}

In the non-uniform version of blockwise decomposability, it is allowed that the
blocks are decomposable over different partitions. This property will
be used to characterize constraints having small FDD representations. 

\begin{definition}[Blockwise decomposability] \label{def:blockwisedecomposable}
  A constraint $R(\vec u)$ is \emph{blockwise decomposable
        in $x,y$} if $M^{R(\vec u)}_{x,y}$ is a proper block matrix with partitions
  $(A_1,\ldots,A_k)$, $(B_1,\ldots,B_k)$ and for each $i\in[k]$ 
  there is a partition $(V_i,W_i)$ of $\tilde u$ with $x\in V_i$ and
  $y\in W_i$
  such that each block $R(\vec{u})|_{x\in A_i,y\in B_i}$
  is decomposable in
  $(V_i,W_i)$.
  A constraint $R(\vec u)$ is \emph{blockwise decomposable}
  if it is decomposable in every pair $x,y \in \tilde u$.
\end{definition}

Note that every uniformly blockwise decomposable relation is also
blockwise decomposable. The next example illustrates that the converse
does not hold.

\begin{example}\label{exa:separateNotions}
  Consider the 4-ary constraint relations
  \begin{align}
        R_1 &:= \{(a,a,a,a),\,(a,b,a,b),\,(b,a,b,a),\,(b,b,b,b)\}\\
        R_2 &:= \{(c,c,c,c),\,(c,d,d,c),\,(d,c,c,d),\,(d,d,d,d)\}\\
        R   &:= R_1\cup R_2
  \end{align}
  Then the selection matrix $M^{R(x,y,u,v)}_{x,y}$ of the constraint $R(x,y,u,v)$
  has two non-empty blocks:
  \begin{align*}
        \label{eq:7}
M_{x,y}^{R(x,y,u,v)}=\left(\begin{array}{c|cccc}
x\backslash y & a & b & c & d\\
\hline a & \left\{ \left(a,a\right)\right\}  & \left\{ \left(a,b\right)\right\}  & \emptyset & \emptyset\\
b & \left\{ \left(b,a\right)\right\}  & \left\{ \left(b,b\right)\right\}  & \emptyset & \emptyset\\
c & \emptyset & \emptyset & \left\{ \left(c,c\right)\right\}  & \left\{ \left(d,c\right)\right\} \\
d & \emptyset & \emptyset & \left\{ \left(c,d\right)\right\}  & \left\{ \left(d,d\right)\right\} 
\end{array}\right)
  \end{align*}
The first block $R(x,y,u,v)|_{x\in\{a,b\},y\in\{a,b\}} = R_1(x,y,u,v)$
is decomposable in $\{x,u\}$ and  $\{y,v\}$. The second block $R(x,y,u,v)|_{x\in\{c,d\},y\in\{c,d\}} = R_2(x,y,u,v)$ is
decomposable in $\{x,v\}$ and $\{y,u\}$. Thus, the constraint is blockwise
decomposable in $x,y$, and we will see in later examples that it is in fact blockwise decomposable in every pair of variables. However, it is not uniformly blockwise decomposable.
\end{example}

Finally, we transfer these characterizations to relations and constraint languages:
  A $k$-ary relation $R$ is (uniformly) blockwise decomposable, if the
  constraint $R(x_1,\ldots, x_k)$ is (uniformly) blockwise decomposable for
  pairwise distinct variables $x_1,\ldots, x_k$.
        A constraint language $\Gamma$ is (uniformly) blockwise decomposable
  if every relation in $\Gamma$ is (uniformly) blockwise
  decomposable.
        A constraint language $\Gamma$ is \emph{strongly (uniformly) blockwise
  decomposable}
  if its co-clone $\langle\Gamma\rangle$ is (uniformly) blockwise
  decomposable.

Now we are ready to formulate our main theorems. The first one states that
the strongly uniformly blockwise decomposable constraint languages are
precisely those that can be efficiently compiled to a structured
representation format (anything between ODDs and structured
DNNFs). 

 \begin{theorem}\label{thm:ODD}
    Let $\Gamma$ be a constraint language.
    \begin{enumerate}
    \item \label{thm:ODD:upper}
    If $\Gamma$ is strongly uniformly blockwise decomposable, then there is a polynomial
    time algorithm that constructs an
    {\upshape ODD} for a given {\upshape CSP($\Gamma$)}-instance.
    \item \label{thm:ODD:lower}
      If $\Gamma$ is not strongly uniformly blockwise decomposable, 
    then there is a family $(I_n)$ of {\upshape
      CSP($\Gamma$)}-instances such that any structured {\upshape DNNF} for $I_n$ has size $2^{\Omega(\|I_n\|)}$.
    \end{enumerate}
\end{theorem}

Our second main theorem states that the larger class of strongly
blockwise decomposable constraint languages captures CSPs that can be
efficiently compiled in an unstructured format between FDDs and DNNFs.

 \begin{theorem}\label{thm:FDD}
    Let $\Gamma$ be a constraint language.
    \begin{enumerate}
    \item \label{thm:FDD:upper}
    If $\Gamma$ is strongly blockwise decomposable, then there is a polynomial
    time algorithm that constructs an
    {\upshape FDD} for a given {\upshape CSP($\Gamma$)}-instance.
    \item \label{thm:FDD:lower}
      If $\Gamma$ is not strongly blockwise decomposable, 
    then there is a family $(I_n)$ of {\upshape
      CSP($\Gamma$)}-instances such that any {\upshape DNNF} for $I_n$ has size $2^{\Omega(\|I_n\|)}$.
    \end{enumerate}
\end{theorem}

Moreover, we will show that both properties (strong blockwise decomposability
and strong uniform blockwise decomposability) are decidable
(Theorem~\ref{thm:decidability}) and that $\Gamma=\{R\}$ for the relation $R$ from
Example~\ref{exa:separateNotions} separates both notions (Theorem~\ref{thm:separation}).

\section{Properties of the Decomposability Notions}\label{sct:properties}

In this section we state and prove important properties of (uniform)
blockwise decomposability.
\subsection{Basic Properties}\label{sec:basic-properties}

We start with some basic properties that will be useful throughout the rest of this paper.

\begin{lemma}\label{lem:closedProjectionSeletion}
        Let $R(\vec u)$ be a blockwise decomposable, resp.~uniformly blockwise decomposable, constraint and let $Y\subseteq
        \tilde u$, $z\in \tilde u$, and $S \subseteq D$. Then the
        projection $\pi_Y R(\vec u)$ as well as the selection $R(\vec
        u)|_{z\in S}$  are also
        blockwise decomposable, resp.~uniformly blockwise decomposable. 
\end{lemma}

\begin{proof}
  If $R(\vec u)$ is (uniformly) blockwise decomposable consider $x,y\in Y$. First note that each entry $M_{x,y}^{\pi_Y R(\vec
        u)}[a_i,b_j]$ is non-empty if and only if $M_{x,y}^{R(\vec u)}[a_i,b_j]$ is
  non-empty. Thus, $M_{x,y}^{\pi_Y R(\vec u)}$ is a proper block
  matrix with the same block structure $(A_1,\ldots,A_k)$ and
  $(B_1,\ldots,B_k)$ as $M_{x,y}^{R(\vec u)}$. Furthermore, if a block $R(\vec u)|_{x\in
        A_i,y\in B_i}$ is decomposable in $(V,W)$, then we have that the corresponding
  block $(\pi_Y R(\vec u))|_{x\in A_i,y\in B_i}$ in $M_{x,y}^{\pi_Y R(\vec u)}$ is decomposable in
  $(V\cap Y,W\cap Y)$. It follows that $\pi_Y R(\vec u)$ is
  (uniformly) blockwise decomposable.

  For the selection ``$z\in S$'' first consider the case where $x=z$ or
  $y=z$. Then the selection matrix $M_{x,y}^{R(\vec u)|_{z\in S}}$ is
  a submatrix of $M_{x,y}^{R(\vec u)}$ and hence  (uniformly)
  blockwise decomposable if $M_{x,y}^{R(\vec u)}$ was (uniformly)
  blockwise decomposable.
  In case $z\notin\{x,y\}$, the matrix $M_{x,y}^{R(\vec u)|_{z\in S}}$
  has the same block structure as $M_{x,y}^{R(\vec u)}$ and its
  entries decomposable w.r.t.~the same partitions. Hence $R(\vec
  u)|_{z\in S}$ is also (uniformly) blockwise decomposable in the
  pair $x,y$.  
\end{proof}

\begin{corollary}\label{cor:unary}
If a constraint language $\Gamma$ is strongly (uniformly) blockwise
decomposable, then its individualization $\Gamma^{\textup{id}} :=
\Gamma\cup \big\{ S : S\subseteq D\big\}$ is also strongly (uniformly) blockwise decomposable.
\end{corollary}
\begin{proof}
  Consider a pp-formula defining 
  $R(\vec{x})\in\langle\Gamma^{\textup{id}}\rangle$ where with Lemma~\ref{lem:closedProjectionSeletion} we may assume that the formula contains no projection. Then the formula can be rewritten as
  $R(\vec{x})= F(\vec x) \land \bigwedge_{i}U_{S_i}(x_i)$, where
  $F(\vec x)$ is a $\Gamma$-formula. Since by assumption
  $F(\vec x)$ is (uniformly) blockwise decomposable, the same holds for
  $R(\vec{x})$ by repeatedly applying the selection $x_i\in S_i$ and Lemma~\ref{lem:closedProjectionSeletion}.
\end{proof}

We next show that when dealing with blockwise decomposable relations, we can essentially assume that they are binary. To this end, we define for every constraint $R(\vec x)$ the set of binary projections 
\begin{align*}
    \Pi_{\le 2}(R(\vec x)) := \{\pi_Y(R(\vec x))| Y\subseteq \tilde{x}, |Y| \le 2\}.
\end{align*}

    \begin{figure}
      \centering
      \begin{tikzpicture}

        \node[minimum size = 4mm, inner sep=0pt, circle] (1d) at (1,0*0.5) {$d$};
        \node[minimum size = 4mm, inner sep=0pt, circle] (1c) at (1,1*0.5) {$c$};
        \node[minimum size = 4mm, inner sep=0pt, circle] (1b) at (1,2*0.5) {$b$};
        \node[minimum size = 4mm, inner sep=0pt, circle] (1a) at (1,3*0.5) {$a$};

        \node[minimum size = 4mm, inner sep=0pt, circle] (2d) at (2,0*0.5) {$d$};
        \node[minimum size = 4mm, inner sep=0pt, circle] (2c) at (2,1*0.5) {$c$};
        \node[minimum size = 4mm, inner sep=0pt, circle] (2b) at (2,2*0.5) {$b$};
        \node[minimum size = 4mm, inner sep=0pt, circle] (2a) at (2,3*0.5) {$a$};

        \draw[-,semithick] (1a) -- (2a);
        \draw[-,semithick] (1b) -- (2b);
        \draw[-,semithick] (1c) -- (2c);
        \draw[-,semithick] (1d) -- (2d);

        \draw[-,semithick] (1c) -- (2d);
        \draw[-,semithick] (1d) -- (2c);

        \node[minimum size = 4mm,inner sep=0pt, circle] (3d) at (3,0*0.5) {$d$};
        \node[minimum size = 4mm,inner sep=0pt, circle] (3c) at (3,1*0.5) {$c$};
        \node[minimum size = 4mm,inner sep=0pt, circle] (3b) at (3,2*0.5) {$b$};
        \node[minimum size = 4mm,inner sep=0pt, circle] (3a) at (3,3*0.5) {$a$};

        \node[minimum size = 4mm,inner sep=0pt, circle] (4d) at (4,0*0.5) {$d$};
        \node[minimum size = 4mm,inner sep=0pt, circle] (4c) at (4,1*0.5) {$c$};
        \node[minimum size = 4mm,inner sep=0pt, circle] (4b) at (4,2*0.5) {$b$};
        \node[minimum size = 4mm,inner sep=0pt, circle] (4a) at (4,3*0.5) {$a$};

        \draw[-,semithick] (3a) -- (4a);
        \draw[-,semithick] (3b) -- (4b);
        \draw[-,semithick] (3c) -- (4c);
        \draw[-,semithick] (3d) -- (4d);

        \draw[-,semithick] (3a) -- (4b);
        \draw[-,semithick] (3b) -- (4a);

        \node[minimum size = 4mm,inner sep=0pt, circle] (5d) at (5,0*0.5) {$d$};
        \node[minimum size = 4mm,inner sep=0pt, circle] (5c) at (5,1*0.5) {$c$};
        \node[minimum size = 4mm,inner sep=0pt, circle] (5b) at (5,2*0.5) {$b$};
        \node[minimum size = 4mm,inner sep=0pt, circle] (5a) at (5,3*0.5) {$a$};

        \node[minimum size = 4mm,inner sep=0pt, circle] (6d) at (6,0*0.5) {$d$};
        \node[minimum size = 4mm,inner sep=0pt, circle] (6c) at (6,1*0.5) {$c$};
        \node[minimum size = 4mm,inner sep=0pt, circle] (6b) at (6,2*0.5) {$b$};
        \node[minimum size = 4mm,inner sep=0pt, circle] (6a) at (6,3*0.5) {$a$};

        \draw[-,semithick] (5a) -- (6a);
        \draw[-,semithick] (5b) -- (6b);
        \draw[-,semithick] (5c) -- (6c);
        \draw[-,semithick] (5d) -- (6d);

        \draw[-,semithick] (5a) -- (6b);
        \draw[-,semithick] (5b) -- (6a);
        \draw[-,semithick] (5c) -- (6d);
        \draw[-,semithick] (5d) -- (6c);

        \node[inner sep=0pt,anchor = south east, align=left] at
        (-.8,-.15) {
          $R = \{(a,a,a,a),$\\
          $\phantom{R = \{}(b,b,b,b),$\\
          $\phantom{R = \{}(c,c,c,c),$\\
          $\phantom{R = \{}(d,d,d,d),$\\
          $\phantom{R = \{}(a,b,a,b),$\\
          $\phantom{R = \{}(b,a,b,a),$\\
          $\phantom{R = \{}(c,d,d,c),$\\
          $\phantom{R = \{}(d,c,c,d)\}$
        };

        \node at (1.5,4*0.5) {$R'$};
        \node at (3.5,4*0.5) {$R''$};
        \node at (5.5,4*0.5) {$R'''$};
      \end{tikzpicture}
      \caption{All binary projections of relation $R$
        (Example~\ref{exa:projections1})}
      \label{fig:bin-projections}
    \end{figure}
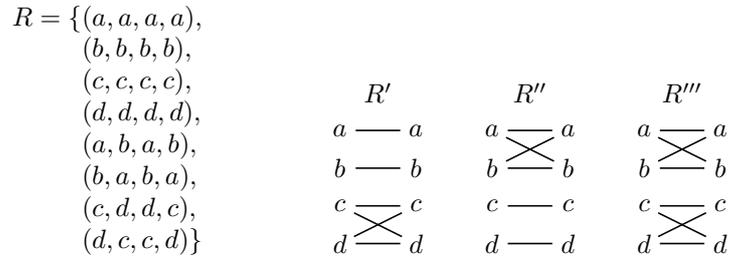

    \begin{figure}
      \centering
      \begin{tikzpicture}
        
        \node[minimum size = 4mm, inner sep=0pt, circle] (xd) at (0,1*0.5) {$d$};
        \node[minimum size = 4mm, inner sep=0pt, circle] (xc) at (0,2*0.5) {$c$};
        \node[minimum size = 4mm, inner sep=0pt, circle] (xb) at (0,3*0.5) {$b$};
        \node[minimum size = 4mm, inner sep=0pt, circle] (xa) at (0,4*0.5) {$a$};
        \node[minimum size = 4mm, inner sep=0pt, draw, rectangle] (x) at (0,5.5*0.5) {$x$};

        \node[minimum size = 4mm, inner sep=0pt, circle] (yd) at (0,-1*0.5) {$d$};
        \node[minimum size = 4mm, inner sep=0pt, circle] (yc) at (0,-2*0.5) {$c$};
        \node[minimum size = 4mm, inner sep=0pt, circle] (yb) at (0,-3*0.5) {$b$};
        \node[minimum size = 4mm, inner sep=0pt, circle] (ya) at (0,-4*0.5) {$a$};
        \node[minimum size = 4mm, inner sep=0pt, draw, rectangle] (y) at (0,-5.5*0.5) {$y$};

        \node[minimum size = 4mm, inner sep=0pt, circle] (vd) at (1*0.5,0) {$d$};
        \node[minimum size = 4mm, inner sep=0pt, circle] (vc) at (2*0.5,0) {$c$};
        \node[minimum size = 4mm, inner sep=0pt, circle] (vb) at (3*0.5,0) {$b$};
        \node[minimum size = 4mm, inner sep=0pt, circle] (va) at (4*0.5,0) {$a$};
        \node[minimum size = 4mm, inner sep=0pt, draw, rectangle] (v) at (5.5*0.5,0) {$v$};

        \node[minimum size = 4mm, inner sep=0pt, circle] (ud) at (-1*0.5,0) {$d$};
        \node[minimum size = 4mm, inner sep=0pt, circle] (uc) at (-2*0.5,0) {$c$};
        \node[minimum size = 4mm, inner sep=0pt, circle] (ub) at (-3*0.5,0) {$b$};
        \node[minimum size = 4mm, inner sep=0pt, circle] (ua) at (-4*0.5,0) {$a$};
        \node[minimum size = 4mm, inner sep=0pt, draw, rectangle] (u) at (-5.5*0.5,0) {$u$};

        \draw[-,semithick] (xa) -- (va);
        \draw[-,semithick] (xb) -- (vb);
        \draw[-,semithick] (xc) -- (vc);
        \draw[-,semithick] (xd) -- (vd);
        \draw[-,semithick] (xa) -- (vb);
        \draw[-,semithick] (xb) -- (va);

        \draw[-,semithick] (va) -- (ya);
        \draw[-,semithick] (vb) -- (yb);
        \draw[-,semithick] (vc) -- (yc);
        \draw[-,semithick] (vd) -- (yd);
        \draw[-,semithick] (vc) -- (yd);
        \draw[-,semithick] (vd) -- (yc);

        \draw[-,semithick] (ya) -- (ua);
        \draw[-,semithick] (yb) -- (ub);
        \draw[-,semithick] (yc) -- (uc);
        \draw[-,semithick] (yd) -- (ud);
        \draw[-,semithick] (ya) -- (ub);
        \draw[-,semithick] (yb) -- (ua);

        \draw[-,semithick] (ua) -- (xa);
        \draw[-,semithick] (ub) -- (xb);
        \draw[-,semithick] (uc) -- (xc);
        \draw[-,semithick] (ud) -- (xd);
        \draw[-,semithick] (uc) -- (xd);
        \draw[-,semithick] (ud) -- (xc);

        \node at (1.3,1.3) {$R''$};
        \node at (1.3,-1.3) {$R'$};
        \node at (-1.3,1.3) {$R'$};
        \node at (-1.3,-1.3) {$R''$};
        
        \node at (0,-3.5) {$R(x,y,u,v)=R'(x,u)\land R''(x,v)\land
          R'(y,v)\land R''(y,u)$};
        
      \end{tikzpicture}
      \caption{Illustration of $R \in \langle\{
        R',R''\}\rangle$ (Example~\ref{exa:projections2}), tuples in
        $R$ correspond to 4-cycles in the figure.}
      \label{fig:bin-definability}
    \end{figure}

\begin{example}\label{exa:projections1}
    Let us compute $\Pi_{\le 2}(R)$ for the 4-ary relation $R$ from
    Example~\ref{exa:separateNotions}, see
    Figure~\ref{fig:bin-projections}. For the unary relations, observe
    that every column in $R$ can take all
    four values $a,b,c,d$, so the only unary relation in $\Pi_{\le 2}(R)$ is
    $\{a,b,c,d\}$. The six possible binary projections lead
    to three different constraint relations: 
    \begin{itemize}
    \item 
    projecting $R(x,y,
    u,v)$ to $\{x,u\}$ and $\{y,v\}$ yields the relation\\ $R'=\{(a,a),(b,b)\}\cup\{c,d\}^2$.
    \item 
    projecting $R(x,y,
    u,v)$ to $\{x,v\}$ and $\{y,u\}$ yields the relation\\ $R''=\{a,b\}^2\cup\{(c,c),(d,d)\}$.
    \item 
    projecting $R(x,y,
    u,v)$ to $\{x,y\}$ and $\{u,v\}$ yields the relation\\
    $R'''=\{a,b\}^2\cup\{c,d\}^2$. \qedhere
    \end{itemize}
\end{example}

For blockwise decomposable constraints, binary projections do not change solutions in the following sense.
\begin{lemma}\label{lem:binarize}
    Let $R(\vec x)$ be a blockwise decomposable constraint. Then 
    \begin{align*}
        R(\vec x) = \bigwedge_{R'(\vec y) \in \Pi_{\le 2}(R(\vec x))} R'(\vec y).
    \end{align*}
\end{lemma}
\begin{proof}
We assume w.\,l.\,o.\,g.~that $\vec x=(x_1,\ldots,x_k)$ contains $k$
distinct variables and let $F(\vec x):=\bigwedge_{R'(\vec y) \in
  \Pi_{\le 2}(R(\vec x))} R'(\vec y)$. First observe that $
\operatorname{sol}(R(\vec x)) \subseteq \operatorname{sol}(F(\vec x))$ by the definition of projections.
For the other direction, we let $\alpha \in \operatorname{sol}(F(\vec
x))$ and need to show $\alpha \in \operatorname{sol}(R(\vec
x))$. For this, we show by induction on $\ell = 2,\ldots, k$:
\begin{equation}
  \label{eq:8}
  \text{For all } I \subseteq  \{x_1,\ldots,x_k\} \text{ with } |I|\leq
  \ell\colon\quad \alpha|_I \in \operatorname{sol}(\pi_I(R(\vec x))).
\end{equation}
Note that the lemma follows from this claim by setting $\ell=k$ and that the base case
$\ell=2$ follows from the definition of $F$. So let $\ell\geq 3$ and
assume that \eqref{eq:8} holds for all $I$ with $|I|<\ell$. Fix an
arbitrary $I$ with $|I|=\ell$, let $z_1,z_2,z_3$ be three variables in $I$,
and set $I_i:=I\setminus\{z_i\}$ for $i\in \{1,2,3\}$. By induction assumption we get     
  $\alpha|_{I_i} \in \operatorname{sol}(\pi_{I_i}(R(\vec x)))$, which
  implies that there are $b_1,b_2,b_3\in D$ such that
$\alpha|_{I_i}\cup\{z_i\mapsto b_i\} \in \operatorname{sol}(\pi_{I}(R(\vec x)))$.
Let $a_i:=\alpha(z_i)$. In order to show that $\alpha \in
\operatorname{sol}(\pi_{I}(R(\vec x)))$, we consider the selection
matrix $M:=M^{\pi_{I}(R(\vec x))}_{z_1,z_2}$. Since we have that $M[a_1,a_2]$,
$M[b_1,a_2]$, and $M[a_1,b_2]$ are non-empty,
these entries lie in the same block, which is decomposable into some
partition $(U,V)$ of $I$ with $z_1\in U$ and $z_2\in V$. Moreover,
$\alpha|_{U\setminus \{z_1\}}\in \operatorname{sol}(\pi_{U\setminus \{z_1\}}(M[a_1,b_2]))$ and $\alpha|_{V\setminus
  \{z_2\}}\in \operatorname{sol}(\pi_{V\setminus
  \{z_2\}}(M[b_1,a_2]))$. Decomposability implies that $\alpha|_{U\setminus
  \{z_1\}}\cup \alpha|_{V\setminus \{z_2\}} \in \operatorname{sol}(M[a_1,a_2])$ which is
the same statement as $\alpha|_{I} \in \operatorname{sol}(\pi_{I}(R(\vec x)))$.  
\end{proof}

Let $\Gamma$ be a constraint language. Then we define $\Pi_{\le 2}(\Gamma)$
to be the constraint language consisting of all relations of arity at
most $2$ that are pp-definable over~$\Gamma$. Equivalently, $\Pi_{\le
  2}(\Gamma)$ consists of all relations that are constraint relations
of constraints in $\Pi_{\le 2}(R(\vec x))$ for a relation $R\in
\coclone{\Gamma}$. Observe that even though $\coclone{\Gamma}$ is
infinite, we have that $\Pi_{\le 2}(\Gamma)$ is finite, since it contains only relations of arity at most two.

\begin{example}\label{exa:projections2}
    Consider again the relation $R$ and its projections $R'$ and $R''$
    from Example~\ref{exa:separateNotions},~\ref{exa:projections1} and
    Figure~\ref{fig:bin-projections}.
    We let $\Gamma:= \{R\}$ and want to compute $\Pi_{\le 2}(\Gamma)$. We
    first observe that instead of $R$ we could use the two binary
    projections $R'$ and $R''$, that is, for $\Gamma':=
    \{R', R''\}$ we have  $\coclone{\Gamma} = \coclone{\Gamma'}$. To
    see this, observe first that we have already seen in
    Example~\ref{exa:projections1} that $R', R''\in \Pi_{\le 2}(R) \subseteq \coclone{\Gamma}$. For the other direction, we just have to express $R$ as a pp-formula in $\Gamma'$ which we do by
    \begin{equation}
R(x,y,u,v):= R'(x,u)\land R''(x,v)\land
          R'(y,v)\land R''(y,u), \quad\text{see Figure~\ref{fig:bin-definability}.}
    \end{equation}
    So to compute $\Pi_{\le 2}(\Gamma)$ we can equivalently analyze $\Pi_{\le 2}(\Gamma')$ which is easier to understand. Note first that the only unary relation that is pp-defiable over $\Gamma'$ is $D:=\{a, b, c, d\}$, so we only have to understand binary relations. To this end, we assign a graph $G_F$ to every $\Gamma'$-formula $F$ where the variables are the variables of $F$ and there is an edge between two variables $x$ and $y$ if there is a constraint $R'(x,y)$ or $R''(x,y)$ in $F$. In the former case, we label the edge with $R'$, in the latter with $R''$. Note that an edge can have both labels. An easy induction shows the following: if two variables $x,y$ are connected by a path whose edges are all labeled with $R'$, then in any solution in which $x$ takes value $a$, $y$ takes value $a$ as well and the same statement is true for $b$. Moreover, if $x$ and $y$ are not connected by such a path, then there is a solution in which $x$ takes value $a$ and $y$ takes value $b$ and vice versa. An analogous statement is true for $c$ and $d$ and paths labeled by $R''$.
    
    From these observations, it follows that the only binary relations
    that are pp-definable over $\Gamma$ and that are not already in
    $\Pi_{\le 2}(R)$ are the equality relation $R_=:=\{(a,a), (b,b), (c,c),
    (d,d)\}$ and the trivial relation $R_{\text{triv}}:=D^2$.
Thus we get
  $\Pi_{\le 2}(\Gamma)=\{D, R_=,R_{\text{triv}}, R', R'', R'''\}$. %
\end{example}
    
The following result shows that for strongly blockwise decomposable constraint languages, we may assume that all relations are of arity at most two.
\begin{proposition}\label{prop:restrictionToBinary}
    Let $\Gamma$ be a strongly blockwise decomposable language. Then $\coclone{\Gamma} = \coclone{\Pi_{\le 2}(\Gamma)}$.
\end{proposition}
\begin{proof}
    We first show that $\coclone{\Pi_{\le 2}(\Gamma)} \subseteq \coclone{\Gamma}$. For this, it suffices to show that $\Pi_{\le 2}(\Gamma) \subseteq \coclone{\Gamma}$. So consider a relation $R\in \Pi_{\le 2}(\Gamma)$. By definition, $R(x,y)$ is in particular pp-definable over $\Gamma$, so $R\in \coclone{\Gamma}$ which shows the first direction of the proof.
    
    For the other direction, is suffices to show that $\Gamma\subseteq \coclone{\Pi_{\le 2}(\Gamma)}$. So let $R\in \Gamma$ of arity $k$. For the constraint $R(x_1, \ldots, x_k)$ we have by Lemma~\ref{lem:binarize} that $R(\vec x) = \bigwedge_{R'(\vec y) \in \Pi_{\le 2}(R(\vec x))} R'(\vec y)$ where $\vec x =(x_1, \ldots, x_k)$. Since the relations $R'$ are all in $\Pi_{\le 2}(\Gamma)$, this shows that $R$ is pp-definable over $\Pi_{\le 2}(\Gamma)$ which proves the claim.
\end{proof}
\subsection{The Relation to Strong Balance}
\label{sec:strong-balance}

Next, we will show that blockwise decomposable relations allow for
efficient counting of solutions by connecting it to the work of Dyer
and Richerby~\cite{DyerR13}. To state their dichotomy theorem, we need the following
definitions. A constraint $R(\vec u)$ is
\emph{balanced} w.r.t.~a partition $\tilde x,\tilde y,\tilde z$ of
$\tilde u$, if the $|D|^{|\tilde x|}\times |D|^{|\tilde y|}$
matrix $M^\#_{\tilde x,\tilde y}$ defined by
$M^\#_{\tilde x,\tilde y}[\alpha,\beta] = \left|\sol(R(\vec
  u)|_{\alpha\cup\beta})\right|$ for $\alpha\in D^{\tilde x}$ and
$\beta\in D^{\tilde y}$ is a proper block matrix, where each block has rank one. A constraint language $\Gamma$ is \emph{strongly
  balanced} if every at least ternary pp-definable constraint is
balanced w.r.t.~every partition $\tilde  x', \tilde  y', \tilde  z'$ of its variables. 

\begin{theorem}[Effective counting dichotomy \cite{DyerR13}]\label{thm:DyerRicherby}
  If $\Gamma$ is strongly balanced, then there is a polynomial time
  algorithm that computes $|\sol(I)|$ for a given
  {\upshape CSP($\Gamma$)}-instance $I$.
  If $\Gamma$ is not strongly balanced, then counting solutions
    for {\upshape CSP($\Gamma$)}-instances is $\#$P-complete.
  Moreover, there is a polynomial time algorithm that decides if a given
    constraint language $\Gamma$
    is strongly balanced.
\end{theorem}

Our next lemma
connects
blockwise decomposability with strong balance and leads to a
number of useful corollaries in the next section. %

\begin{lemma}\label{lem:blockisbalanced}
  Every strongly blockwise decomposable constraint language is
  strongly balanced.
\end{lemma}
We first sketch the proof for the case
when $R(\vec u)$ is ternary, before proving the general case:
  Let $\Gamma$ be strongly blockwise decomposable and $R(x,y,z)$ a
  pp-defined ternary constraint. Then the selection matrix  $M^{R(x,y,z)}_{x,y}$ is a block
  matrix, where each block $(A,B)$ is decomposable in some
  $(V,W)$, either $(\{x,z\},\{y\})$ or $(\{x\},\{y,z\})$. In any case, for the
  corresponding block $(A,B)$ in
  $M^\#_{x,y}$ we have
  $M^\#_{x,y}[a,b]=s_a\cdot t_b$ where $s_a=|\operatorname{sol}(\pi_V(R(x,y,z)|_{x=a,y\in
    B}))|$ and $t_b=|\operatorname{sol}(\pi_W(R(x,y,z)|_{x\in
    A,y=b}))|$. Thus, the block has rank 1.

To prove the general case of Lemma~\ref{lem:blockisbalanced}, it will be convenient to
have a generalization of the selection matrix. For a
constraint $R(\vec u)$, partition $\tilde x,\tilde y,\tilde z$ of
$\tilde u$ and domain $D$ we let
$M_{\tilde x, \tilde y}^{R(\vec u)}$ be the
$(|D|^{|\tilde x|})\times (|D|^{|\tilde y|})$-matrix whose rows are
indexed by assignments $\alpha:\tilde x\rightarrow D$, whose columns
are indexed by the assignments $\beta:\tilde y\rightarrow D$ and that
have as entry at position $\alpha,\beta$ the constraint $R(\vec
u)|_{\alpha\cup \beta}$.
We say that $R(\vec u)$ is \emph{blockwise set-decomposable}
with respect to $\tilde x$ and $\tilde y$ if $M_{\tilde x, \tilde
  y}^{R(\vec u)}$ is a proper block matrix and for every
non-empty block $(A,B)$ with $A\subseteq D^{\tilde x}$ and $B\subseteq
D^{\tilde y}$, the selection $R(\vec u)|_{\alpha\in A,\beta\in B}$ is
decomposable such that no factor contains variables of both $\tilde x$ and
$\tilde y$. We say that $R(\vec u)$ is blockwise set-decomposable if it
is blockwise set-decomposable for all choices $\tilde x, \tilde y$ of
disjoint variable sets from $\tilde u$.

\begin{lemma}\label{lem:set-decomposable}
Let $R(\vec u)$ be a constraint. Then $R(\vec u)$ is blockwise set-decomposable if and only if it is blockwise decomposable.
\end{lemma}

\begin{proof}
                If $R(\vec u)$ is blockwise set-decomposable, then it is by definition also blockwise decomposable. It only remains to show the other direction. So assume that $R(\vec u)$ is blockwise decomposable. We proceed by induction on $|\tilde x|+|\tilde y|$. If $|\tilde x|+|\tilde y| = 2$, then $\tilde x$ and $\tilde y$ each consist of one variable, so there is nothing to show.
                
                Now assume that w.l.o.g.~$|\tilde x|>1$. Let $\tilde z$ be
                the variables of $R$ not in $\tilde x\cup \tilde y$. We
                first show that for all choices of $\tilde x, \tilde y$
                the matrix $M^{R(\vec u)}_{\tilde x, \tilde y}$ is a proper block
                matrix with the criterion of Lemma~\ref{lem: 2x2
                  submatrix}. 

               Consider assignments $\alpha_0,\alpha_1\colon \tilde
               x\to D$ and $\beta_0,\beta_1\colon \tilde
               y\to D$ with 
               $R(\vec u)|_{\alpha_0\cup\beta_0}\neq\emptyset$,
               $R(\vec u)|_{\alpha_0\cup\beta_1}\neq\emptyset$, and
               $R(\vec u)|_{\alpha_1\cup\beta_0}\neq\emptyset$. We
               have to show that 
               $R(\vec u)|_{\alpha_1\cup\beta_1}\neq\emptyset$ as well.

                Let $z$ be one of the variables in
                $\tilde x$ and let
                $\alpha'_i:=\alpha_i|_{\tilde x \setminus \{z\}}$
                be the assignment restricted to the other variables
                in $\tilde x$. We have
               $R(\vec u)|_{\alpha'_0\cup\beta_0}\neq\emptyset$,
               $R(\vec u)|_{\alpha'_0\cup\beta_1}\neq\emptyset$, and
               $R(\vec u)|_{\alpha'_1\cup\beta_0}\neq\emptyset$. So by induction,
               $R(\vec u)|_{\alpha'_1\cup\beta_1}\neq\emptyset$ as well.
               It follows that in $M^{R(\vec u)}_{\tilde x
                  \setminus \{z\}, \tilde y}$ all four entries lie in
                a block. From blockwise set-decomposability we get, using the induction assumption, that there are non-empty
                constraints
                $C_{\alpha'_0}(\vec v),C_{\alpha'_1}(\vec
                v),C_{\beta_0}(\vec w),C_{\beta_1}(\vec w)$ with $\tilde x\subseteq \tilde v$ and $\tilde y \subseteq \tilde w$
                such that for all $i,j\in \{0,1\}$ we have $\pi_{\tilde z\cup \{z\}}(R(\vec
                u)|_{\alpha'_i,\beta_j})=C_{\alpha'_i}(\vec v)\times
                C_{\beta_j}(\vec w)$. We now have two cases.

                \paragraph*{Case 1 ($z\in\tilde{v}$):} Since  $R(\vec
                u)|_{\alpha_1\cup\beta_0}\neq\emptyset$ we have that
                $\alpha_1|_z\in \sol(\pi_z R(\vec
                u)|_{\alpha'_1\cup\beta_0}).$ Applying the definitions
                we can therefore conclude that $\emptyset\neq \pi_z R(\vec
                u)|_{\alpha'_1\cup\beta_0}=\pi_z
                C_{\alpha'_1}(\vec
                v)=\pi_z(
                C_{\alpha'_1}(\vec v)\times C_{\beta_1}(\vec
                w))= \pi_z R(\vec
                u)|_{\alpha'_1\cup\beta_1}$, which proves the claim.
                \paragraph*{Case 2 ($z\in\tilde{w}$):} We first assume
                that $\alpha_1|_z\in \sol(\pi_z C_{\beta_1}(\vec w))$
                and will later argue that this has to be the case.
                In this situation we can conclude, similarly to Case 1, that
                $\emptyset\neq \pi_z
                C_{\beta_1}(\vec
                w)=\pi_z(
                C_{\alpha'_1}(\vec v)\times C_{\beta_1}(\vec
                w))= \pi_z R(\vec
                u)|_{\alpha'_1\cup\beta_1}$ which is what we needed to show.
                It remains to show that $\alpha_1|_z\in \sol(\pi_z C_{\beta_1}(\vec w))$. So let us assume for contradiction that $\alpha_1|_z\notin
                \sol(\pi_z C_{\beta_1}(\vec w))$ (but still
                $z\in\tilde w$). This implies in
                particular that $\sol(R(\vec
                u)|_{\beta_1})=\bigcup_{\alpha}\sol(D_\alpha(\vec
                v)\times C_{\beta_1})=\emptyset$
                (for suitable
                constraints $D_\alpha$ occurring in the same block) and hence $R(\vec
                u)|_{\alpha_1|_z\cup\beta_1}=\emptyset$.
                However, as $R(\vec u)|_{\alpha_0\cup\beta_0}$,
               $R(\vec u)|_{\alpha_0\cup\beta_1}$, and
               $R(\vec u)|_{\alpha_1\cup\beta_0}$ are nonempty, the
               same holds for $R(\vec u)|_{\alpha_0|_z\cup\beta_0}$,
               $R(\vec u)|_{\alpha_0|_z\cup\beta_1}$, and
               $R(\vec u)|_{\alpha_1|_z\cup\beta_0}$ which contradicts
               the block structure of $M^{R(\vec u)}_{\{z\}, \tilde
                 y}$ that we are guranteed to have by induction assumption.
                
                Now that we have proven that $M^{R(\vec u)}_{\tilde x,
                  \tilde y}$ is a proper block matrix, we need to
                verify that each block $(A,B)$ decomposes. This follows again
                by considering the matrix $M^{R(\vec u)}_{\tilde x\setminus\{z\},
                  \tilde y}$. By induction assumption the
                corresponding block $(A',B)$ for $A':=A|_{\tilde
                  x\setminus\{z\}}$ decomposes into $D_{\alpha'}\times
                D_\beta$ for $\alpha' \in A'$ and $\beta\in B$. If, as
                in Case~1, $z$ is in the scope of all $D_{\alpha'}$, then
                the original block $(A,B)$ decomposes into
                $(D_{\alpha|_{\tilde x \setminus \{z\}}}\land (z=\alpha(z)))\times
                D_\beta$ for all $\alpha\in A$ and $\beta \in B$. If
                as in Case 2 $z$ is in the domain of all $D_\beta$,
                then $(A,B)$ decomposes into $D_\alpha \times
                D_\beta|_{z=\alpha(z)}$ for all $\alpha\in A$ and $\beta \in B$.
\end{proof}

The following connection between balance and blockwise decomposability is now easy to show.

\begin{lemma}\label{lem:BDimpliesBalanced}
                Every blockwise decomposable constraint is balanced.
\end{lemma}
\begin{proof}
                Let $R(\vec u)$ be blockwise decomposable. Let $\tilde
                x, \tilde y$ be disjoint variable sets. We want to
                show that
                $M^\#_{\tilde x, \tilde y}$
                is
                balanced. Let $(A,B)$ be a block. Since, by
                Lemma~\ref{lem:set-decomposable} the constraint
                $R(\vec u)$ is blockwise set-decomposable, every entry
                in bock $(A,B)$ can be written as
                $|R_{\alpha\cup\beta}| = |C_{\alpha}\times C_{\beta}|$
                where $\alpha\colon \tilde x\to D$, $\beta\colon
                \tilde y\to D$ and the constraints $C_{\alpha}$ and
                $C_{\beta}$ are given by the decomposability of the
                block $(A,B)$. It follows directly that the block is a rank-$1$ matrix, which shows the lemma.
\end{proof}

Lemma~\ref{lem:blockisbalanced} is now a direct consequence of Lemma~\ref{lem:BDimpliesBalanced}.

\subsection{Consequences of the Relation to Strong Balance}

In this section, we will use Lemma~\ref{lem:blockisbalanced} to derive useful properties of strongly blockwise decomposable languages.

\begin{corollary}\label{cor:counting-block-factors}
  Let $\Gamma$ be a strongly blockwise decomposable constraint
  language.
  \begin{enumerate}
  \item\label{enum:counting}
    Given a
    $\Gamma$-formula $\varphi(\vec u)$ and a (possibly empty) partial
    assignment $\alpha$, the number $\big|\operatorname{sol}(F_\varphi(\vec u)|_\alpha)\big|$ of solutions that extend
    $\alpha$ can be computed in polynomial time.
  \item\label{enum:block} Given a
    pp-formula $\varphi(\vec u)$ over $\Gamma$ and $x,y\in \tilde u$, the blocks
    of $M^{F_\varphi(\vec u)}_{x,y}$ can be computed in polynomial time.
  \item\label{enum:factors} Given a
    pp-formula $\varphi(\vec u)$ over $\Gamma$, the partition $(C_1,\ldots,C_k)$ of $\tilde u$
    into indecomposable factors of $F_\varphi(\vec u)$ can be computed in polynomial time.
  \end{enumerate}
\end{corollary}

  \begin{proof}
  Claim~\ref{enum:counting} follows immediately from the combination of
  Lemma~\ref{lem:blockisbalanced} with Theorem~\ref{thm:DyerRicherby}
  and the fact that strongly blockwise decomposable constraint
  languages are closed under selection
  (Corollary~\ref{cor:unary}). For Claim~\ref{enum:block} let $\varphi(\vec
  u)=\pi_{\tilde u}(\psi(\vec v))$ for some $\Gamma$-formula $\psi(\vec
  v)$. To compute the blocks of $M^{F_\varphi(\vec
    u)}_{x,y}$, we can use Claim~\ref{enum:counting} to compute for every $x=a$ and
  $y=b$ whether $|\sol(M^{F_\psi(\vec v)}_{x,y}[a,b])|>0$ and hence $M^{F_\varphi(\vec
    u)}_{x,y}\neq \emptyset$.

  To prove Claim~\ref{enum:factors}, we abbreviate $F(\vec
  u):= F_\varphi(\vec u)$ and note that by Claim~\ref{enum:block} we can
  calculate the block structure of $M_{x,y}^{F(\vec u)}$ for every variable pair
  $x,y \in \tilde u$. Consider the graph $G$ with $V(G)=\tilde u$
  and edges between $x$ and $y$ if $M_{x,y}^{F(\vec u)}$ has at least
  two non-empty blocks. We claim that if $\{x,y\}\in E(G)$ for some $x$ and $y$,
  then $x$ and $y$ must appear in the same indecomposable factor of
  ${F(\vec u)}$. To see this, assume by way of contradiction that $x$ and $y$ appear in two different factors of $F(\vec u)$. Then there are disjoint sets $\tilde x, \tilde y\subseteq \tilde u$ with $x\in \tilde x$ and $y \in \tilde y$ such that $F(\vec u) = \pi_{\tilde x}(F(\vec u)) \times \pi_{\tilde y}(F(\vec u))$. Let $A:= \{d \in D\mid \exists \alpha \in \sol(F(\vec u)): \alpha(x) = d\}$ and $B:= \{d \in D\mid \exists \beta \in \sol(F(\vec u)): \beta(y) = d\}$. Then $(A,B)$ is the only block of $M^{F(\vec u)}_{x,y}$, so by definition there is no edge between $x$ and $y$, which proves the claim.
  
  Let  $C_1, \ldots, C_k$ be the connected components of $G$. All variables of one connected component must appear in the same factor, so $\pi_{C_i}(F(\vec u))$ is indecomposable. We claim that
$
{F(\vec u)} = \pi_{C_1}({F(\vec u)}) \times \cdots \times \pi_{C_k}({F(\vec u)})
$.
It suffices to show ${F(\vec u)} = \pi_{C}({F(\vec u)})\times
\pi_{V(G)\setminus C}({F(\vec u)})$ for one connected component $C$,
since this can be used iteratively to show the claim. We have that for
any $x\in C$ and $y\in V(G)\setminus C$, the selection matrix
$M_{x,y}^{F(\vec u)}$ has only one block, which can be decomposed into
$(\tilde v, \tilde w)$, with $x\in \tilde v$ and $y\in \tilde w$, so that 
\[
F(\vec u) = \pi_{\tilde v}(F(\vec u)) \times \pi_{\tilde w}(F(\vec u)).
\]
Since $x\in\tilde v$, no variable from the same connected component
$C$ can be in $\tilde w$ and therefore 
$\tilde v \supseteq C$.
Since we can obtain such a decomposition for every $x\in C$ and $y\in
V(G)\setminus C$, and the intersection of all these decompositions
induces a decomposition where $C$ is separated from any $y$, we get
that ${F(\vec u)} = \pi_{C}({F(\vec u)})\times
\pi_{V(G)\setminus C}({F(\vec u)})$.
\end{proof}

We close this section by stating the following property that applies only
to \emph{uniformly} decomposable constraints.

\begin{lemma}\label{lem:onlycheckprojections}
        Let $R(\vec u)$ be a constraint that is uniformly blockwise decomposable in $x$ and $y$. Then there exist $\vec v$ and $\vec w$ with $\tilde u = \{x,y\} \mathrel{\dot \cup} \tilde v \mathrel{\dot \cup} \tilde w$ such that 
        \begin{align}\label{eq:onlycheckproj}
                R(\vec u) =  \pi_{x,\tilde v} R(\vec u) \wedge \pi_{y,\tilde w} R(\vec u) \wedge \pi_{x,y} R(\vec u).
        \end{align}
        Furthermore, if $R(\vec u)$ is defined by a pp-formula $\varphi(\vec u)$ over a strongly
        uniformly blockwise decomposable $\Gamma$, then $\vec v$ and
        $\vec w$ can be computed from $\varphi(\vec u)$ in polynomial time.
\end{lemma}
\begin{proof}
        Since $R(\vec u)$ is uniformly blockwise decomposable in $x$ and $y$, there exist $\vec v$ and $\vec w$ with $\tilde u = \{x,y\} \mathrel{\dot \cup} \tilde v \mathrel{\dot \cup} \tilde w$ such that for all $a,b \in D$ we have
  \begin{align}
   M^{R(\vec u)}_{x,y}[a,b] = \pi_{\tilde v}(R(\vec u)|_{x=a}) \times \pi_{\tilde w}(R(\vec u)|_{y=b})  \; \text{ or } \; \operatorname{sol}(M^{R(\vec u)}_{x,y}[a,b])
        = \emptyset.
  \end{align}
        We will show (\ref{eq:onlycheckproj}) for this choice of
        $\tilde v$ and $\tilde w$. First, if $\beta: \tilde u \to D$
        satisfies $R(\vec u)$, then by definition of projections, we
        have for every $\tilde x\subseteq \tilde u$ that
        $\beta|_{\tilde x}$ satisfies $\pi_{\tilde x}(R(\vec u))$,
        which yields the containment $\subseteq$ of
        (\ref{eq:onlycheckproj}). Now let $\alpha \in  \pi_{x,\tilde
          v} R(\vec u) \wedge \pi_{y,\tilde w} R(\vec u) \wedge
        \pi_{x,y} R(\vec u)$. Since $\alpha |_{\{x,y\}} \in
        \pi_{x,y}R(\vec u)$ we get $M_{x,y}^{R(\vec u)}[\alpha(x),\alpha(y)] \neq \emptyset$ and thus
\begin{align}
                M_{x,y}^{R(\vec u)}[\alpha(x),\alpha(y)] = \pi_{\tilde v}(R(\vec u)|_{x=\alpha(x)}) \times \pi_{\tilde w}(R(\vec u)|_{y=\alpha(y)}).
\end{align}     
Now
$\alpha |_{\{x\}\cup \tilde v} \in \pi_{\{x\}\cup \tilde v} (R(\vec
u))$ implies that
$\alpha |_{\tilde v} \in \pi_{\tilde v} \left(R(\vec u
  )|_{x=a}\right)$. Analogously,
$\alpha |_{\{y\}\cup \tilde w} \in \pi_{\{y\}\cup \tilde  w}
R(\vec u)$ implies that
$\alpha |_{\tilde w} \in \pi_{\tilde w} \left(R(\vec u)
  |_{y=b}\right)$. Thus we get
$\alpha |_{\tilde v\cup \tilde w} \in M_{x,y}^{R(\vec u)}[\alpha(x),\alpha(y)]$
and $\alpha \in R(\vec u)$.

We now show how to find $\vec v$ and $\vec w$ in time polynomial
in the size of the pp-definition $\varphi(\vec u)$ of $R(\vec u)$. 
To this end, we first compute the block structure of $M^{R(\vec
  u)}_{x,y}$ with
Corollary~\ref{cor:counting-block-factors}.\ref{enum:block}. Let
$(A_1, B_1), \ldots , (A_\ell, B_\ell)$ these blocks. Then we compute
for every $i\in [\ell]$ the indecomposable factors of the block $(A_i,
B_i)$ by applying
Corollary~\ref{cor:counting-block-factors}.\ref{enum:factors} to the
formula $\varphi_i(\vec u):= \varphi(\vec u)\land U_{A_i}(x) \land U_{B_i}(y)$ which we can do by Corollary~\ref{cor:unary}. Denote for every $i\in [\ell]$ the corresponding variable partition by~$\rho_i$.

It remains to compute a variable set $\tilde v^*$ with $x\in \tilde v^*$ such that for all $i\in [\ell]$ we have that $\pi_{\tilde v^*}F_{\varphi_i}(\vec u)$ is a factor of $F_{\varphi_i}$ and $y\notin \tilde v^*$. To this end, observe that if there is an $i\in [\ell]$ and a set $\tilde v'\in \rho_i$ that contains two variables $x',y'$, then either $x'$ and $y'$ must both be in $\tilde v^*$ or they must both be in $\tilde u\setminus \tilde v^*$. This suggests the following algorithm: initialize $\tilde v^*:= \{x\}$. While there is an $i\in [\ell]$ and a set $\tilde v'\in \rho_i$ such that there are variables $x'\in \tilde v'\cap \tilde v^*$ and $y'\in \tilde v'\setminus \tilde v^*$, add $y'$ to $\tilde v^*$. We claim that when the loop stops, we have a set $\tilde v^*$ with the desired properties. First, observe that we have for all $i\in [\ell]$ that $\pi_{\tilde v^*}F_{\varphi_i}(\vec u)$ is a factor of $F_{\varphi_i}(\vec u)$, because otherwise we would have continued adding elements. Moreover, $x\in \tilde v^*$ by construction. Finally, since a decomposition with the desired properties exists by what we have shown above, the algorithm will never be forced to add $y$ to $\tilde v^*$. This proves the claim and thus complete the proof of the lemma.
\end{proof}

\section{Algorithms}\label{sct:algorithms}

\subsection{Polynomial time construction of ODDs for strongly uniformly blockwise decomposable constraint languages}\label{subsec:ODDalgo}

The key to the efficient construction of ODD for uniformly
blockwise decomposable constraints is the following lemma, which
states that any such constraint is equivalent to a treelike conjunction of binary projections. 

\begin{lemma}[Tree structure lemma]\label{lem:thereisatree}
        Let $R(\vec u)$ be a constraint that is
        uniformly blockwise decomposable. Then there is an undirected tree
        $T$ with vertex set $V(T) = \tilde u$ such that 
        \begin{align*}
                R(\vec u) =  \textstyle\bigwedge_{\{p,q\} \in E(T)}
          \pi_{\{p,q\}}(R(\vec u)).
        \end{align*}
        Furthermore, $T$ can be calculated in polynomial
        time from a pp-definition $\varphi(\vec u)$ of $R(\vec u)$ over a strongly uniformly blockwise decomposable language $\Gamma$.
\end{lemma}
\begin{proof}
        We first fix $x$ and $y$ arbitrarily and apply
        Lemma~\ref{lem:onlycheckprojections} to obtain a tri-partition ($\tilde v$, $\{x,y\}$, $\tilde w$) of
        $\tilde u$ such that $R(\vec u) =  \pi_{x,\tilde v}(R(\vec
        u))\wedge \pi_{x,y}(R(\vec u)) \wedge \pi_{y,\tilde w}(R(\vec
        u)) $. We add the edge $\{x,y\}$ to $T$.
        By Lemma~\ref{lem:closedProjectionSeletion},
        $\pi_{x,\tilde v}(R(\vec u))$ and $\pi_{y,\tilde w}(R(\vec u))$ are uniformly blockwise 
        decomposable, so Lemma~\ref{lem:onlycheckprojections}
        can be recursively applied on both projections. For (say)
        $\pi_{x,\vec v}(R(\vec u))$ we fix $x$, choose an arbitrary $z \in \tilde
        v$, apply Lemma~\ref{lem:onlycheckprojections}, and add the
        edge $\{x,z\}$ to $T$. Continuing this
        construction recursively until no projections with more than
        two variables are left yields the desired result.
        
        Since for every edge in $T$, we only have to make one
        recursive call in the above algorithm, the total number of
        recursive calls is bounded by the number of
        variables. Moreover, if $R(\vec u)$ is defined by a pp-formula
        $\varphi(\vec u)$ over a strongly uniformly blockwise decomposable
        language, then the applications of
        Lemma~\ref{lem:onlycheckprojections} take only polynomial time
        in the size of $\varphi(\vec u)$. All other operations in each recursive call also clearly only take polynomial time, so overall the construction can be done in polynomial time.
\end{proof}

From the tree structure of Lemma~\ref{lem:thereisatree}, we will construct small ODDs by starting with
a  centroid, i.\,e., a variable whose removal splits
    the tree into connected components of at most $n/2$
    vertices each where $n$ is the number of vertices in the tree, which in our case is the number of variables. From the tree structure lemma it follows that we can
    handle the (projection on the) subtrees independently. A recursive
    application of this idea leads to an ODD of size $O(n^{\log |D|+1})$.%

    \begin{proof}[Proof of part \ref{thm:ODD:upper} in Theorem~\ref{thm:ODD}]
      Let $I$ be a CSP($\Gamma$)-instance and $\varphi_I(\vec u)$ the
      corresponding $\Gamma$-formula. By Lemma~\ref{lem:thereisatree}
      we can compute in polynomial time a tree $T$ such that
      $F_{\varphi_I}(\vec u) = \bigwedge_{\{p,q\} \in E(T)}
      \pi_{\{p,q\}}(F_{\varphi_I}(\vec u))$. By Corollary~\ref{cor:counting-block-factors}.\ref{enum:counting} we can explicitly
      compute, for each $\{p,q\} \in E(T)$, a binary relation
      $R_{\{p,q\}}\subseteq D^2$ such that
      $R_{\{p,q\}}(p,q) = \pi_{\{p,q\}}(F_{\varphi_I}(\vec u))$. Now we define
      the formula
      $\psi_T(\vec u):=\bigwedge_{\{p,q\} \in E(T)}
      R_{\{p,q\}}(p,q)$
      and
      note that $\sol(F_{\psi_T}(\vec u))=\sol(F_{\varphi_I}(\vec u))$. It remains to
      show that such tree-CSP instances can be efficiently compiled to
      ODDs. This follows from the following inductive claim, where for
      technical reasons we also add unary constraints $U_{D_v}(v)$ for
      each vertex $v$ (setting $D_v := D$ implies the theorem).

      \begin{claim*}
        Let $T'$ be a tree on $n$ vertices and $$\psi_{T'}(\vec u')=\bigwedge_{\{v,w\}
        \in E(T')} R'_{\{v,w\}}(v,w) \land \bigwedge_{v\in\tilde
        u} U_{D_v}(v)$$ be a formula. Then there is an order $<$,
      depending only on $T'$, such that an ODD$_<$ of size at most $f(n) :=  n|D|^{\log(n)}$ representing $F_{\psi_{T'}}(\vec u')$ can be computed
      in $n^{O(1)}$.
      \end{claim*}

We prove the claim by induction on $n$. The case $n=1$ is
trivial. 
If $n\geq 2$ let~$z$ be
    a centroid in $T'$, that is a node whose removal splits the tree into
    $\ell\geq 1$
    connected components $T'_1$,\ldots, $T'_\ell$ of at most $n/2$ vertices each. It is a classical result that every tree has at least one centroid~\cite{Jordan1869}. Let $\vec v_1$, \ldots, $\vec v_\ell$ be vectors of the variables
    in these components, so ($\{z\}$, $\tilde v_1$, \ldots,
    $\tilde v_\ell$) partitions $\tilde u'$.
    Let $x_i\in V(T'_i)$ be the neighbors of $z$ in $T'$. We want to branch on $z$ and recurse on the connected components
    $T'_i$. To this end, for each assignment $z\mapsto a$ we remove for each
    neighbor $x_i$ those values that cannot be extended to $z\mapsto a$. That
    is, $D^a_{x_i}:= \{b \in D \mid \{x_i\mapsto b, z\mapsto
    a\} \in \sol(U_{D_{x_i}}(x_i)\land R'_{\{x_i,z\}}(x_i,z)\land U_{D_{z}}(z))\}$.
    Now we let $\psi^a_i(\vec v_i) := \bigwedge_{\{v,w\}
        \in E(T'_i)} R'_{\{v,w\}}(v,w) \land \bigwedge_{v\in
        \tilde v_i\setminus \{x_i\}} U_{D_v}(v)\land U_{D^a_{x_i}}(x_i)$
      and observe that
      \begin{equation}
         \operatorname{sol}(F_{\psi_{T'}}(\vec u')) = \dot\bigcup_{a\in D}\operatorname{sol}(U_a(z)\times F_{\psi^a_1}(\vec
         v_1)\times\cdots \times F_{\psi^a_\ell}(\vec v_\ell)).
      \end{equation}
    By induction assumption, for each $i\in[\ell]$ there is an order
    $<_i$ of $\tilde v_i$ such that each $F_{\psi^a_i}(\vec v_i)$ has an
    ODD$^a_{<_i}$ of size $f(n_i)$ for $n_i:=|\tilde v_i|$.
    Now we start our ODD for $F_{ T'}(\vec u')$ with branching on $z$
    followed by the sequential combination of ODD$^a_{<_1}$, \ldots, ODD$^a_{<_\ell}$ for each assignment
    $a\in D$ to~$z$. This completes the inductive construction. Since
    its size is bounded by
    $1+|D|\sum_{i\in[\ell]}f(n_i)$, the following easy
    estimations finish the proof of the claim (recall that $\sum_{i\in[\ell]}
  n_i = n-1$):
\begin{align*}
   \label{eq:11}
  1+|D|\textstyle\sum_{i\in[\ell]}f(n_i)
                               &\le  1+|D|\textstyle\sum_{i\in[\ell]}
  n_i |D|^{\log(n/2)} \\
      &=  1+ |D|^{\log(n)} (n-1)  \\
  &\le n |D|^{\log(n)} = n^{\log(|D|)+1} = n^{O(1)}   %
\end{align*}
    \end{proof}

\subsection{Polynomial time construction of FDDs for strongly blockwise decomposable constraint languages}\label{subsec:FDDalgo}

For blockwise decomposable constraints that are \emph{not} uniformly blockwise
decomposable, a good variable order may depend on the values assigned
to variables that are already chosen, so it is not surprising that the
tree approach for ODDs does not work in this setting.

For the construction of the FDD, we first compute the indecomposable
factors (this can be done by Corollary~\ref{cor:counting-block-factors}.\ref{enum:factors}) and treat them
independently. This, of course, could have also been done for the ODD
construction. The key point now is how we treat the indecomposable
factors: every selection matrix $M^{R(\vec x)}_{x,y}$ for a (blockwise decomposable) indecomposable constraint
necessarily has two non-empty blocks. But then every row
$x=a$ must have at least one empty entry
$\operatorname{sol}(M^{R(\vec x)}_{x,y}[a,b])=\emptyset$. This in turn implies
that, once we have chosen $x=a$, we can exclude $b$ as a possible
value for $y$. As we have chosen $y$ arbitrarily, this also applies to
any other variable (maybe with a different domain element $b$). So the set of possible values for every variable
shrinks by one and since the
domain is finite, this cannot happen too
often. Algorithm~\ref{alg:FDD} formalizes this recursive idea. To
bound the runtime of this algorithm, we analyze the size of the
recursion tree.%

\begin{algorithm}[t]
  \begin{algorithmic}[1]
    \Statex \textbf{Input:} $\Gamma$-formula $\varphi(x_1,\ldots,x_n)$
    for a strongly
    blockwise decomposable $\Gamma$ over domain $D$.
    \Statex \textbf{Output:} An FDD deciding $F_\varphi(x_1,\ldots,x_n)$.
    \State Initialize variable domains $D_{x_i}\gets D$ for $i=1,\ldots, n$.
    \State \Return
    \Call{ConstructFDD}{$\varphi(x_1,\ldots,x_n)$; $D_{x_1}$,\ldots,$D_{x_n}$}
    \Statex
    \Procedure{ConstructFDD}{$\psi(x_1,\ldots,x_n)$;
      $D_{x_1}$,\ldots,$D_{x_n}$}
    \Statex\Comment{Receives a pp-definition $\psi$ and variable
      domains $D_{x_i}$.}
    \State $\psi(x_1,\ldots,x_n) \gets \psi(x_1,\ldots,x_n)\land
    \bigwedge_{i\in [n]} U_{D_{x_i}}(x_i)$
    \If{$n=1$}
    \State \Return 1-node FDD deciding $F_\psi(x_1)$, branching on all
    values $D_{x_1}$.
    \EndIf
    \State Compute the partition $C_1, \ldots,
    C_m$ into indecomposable factors of $F_{\psi}(x_1,\ldots,x_n)$. \label{line:indecfac}
    \If{$m\geq 2$}
       \For{$i=1\ldots m$}
       FDD$_i \gets$ \Call{ConstructFDD}{$\pi_{C_i}\psi$; $D_{y}$ for
          $y\in C_i$}\label{line:recursivebecausedecompose}
       \EndFor
       \State \Return Sequential composition of FDD$_1$, \ldots, FDD$_m$ 
    \Else
    \State Introduce branching node for $x_1$.
    \For{$a\in D_{x_1}$}
      \For{$i=2,\ldots,n$}
      \For{$b\in D_{x_i}$}
        \If{$\operatorname{sol}(M^{F_\psi(x_1,\ldots,x_n)}_{x_1,x_i}[a,b]) = \emptyset$}\label{line:testmatrixentry}
         \State $D_{x_i}\gets D_{x_i}\setminus \{b\}$
         \Statex\Comment{This happens for at least one $b\in D_{x_i}$.}
        \EndIf
      \EndFor
      \EndFor
      \State $\psi_a(x_2,\ldots,x_n)\gets
      \pi_{\{x_2,\ldots,x_n\}}(\psi(x_1,\ldots,x_n)\land U_{\{a\}}(x_1))$
      \State FDD$^a \gets $ \Call{ConstructFDD}{$\psi_a(x_2,\ldots,x_n)$; $D_{x_2}$,\ldots,$D_{x_n}$}\label{line:recursebecauseshrinking}
    \EndFor
    \State Connect $x\stackrel{a}{\rightarrow} $ FDD$^a$ for all $a\in
    D_{x_1}$ and \Return resulting FDD
    \EndIf
    \EndProcedure
  \end{algorithmic}
  \caption{FDD construction algorithm}
  \label{alg:FDD}
\end{algorithm}

\begin{proof}[Proof of Item \ref{thm:FDD:upper} of Theorem~\ref{thm:FDD}]
  The algorithm is formalized in Algorithm~\ref{alg:FDD}. It is
  straightforward to verify that this algorithm computes an FDD that
  decides $F_\varphi(x_1,\ldots,x_n)$. It remains to discuss the size of the
  FDD and the running time. First note that the decomposition into
  indecomposable factors (Line~\ref{line:indecfac}) can be computed in polynomial time by
  Corollary~\ref{cor:counting-block-factors}.\ref{enum:factors}. Moreover, (non-)emptiness of the entries of the selection
  matrices (Line~\ref{line:testmatrixentry}) can be tested in
  polynomial time by
  Corollary~\ref{cor:counting-block-factors}.\ref{enum:counting}. %
  Hence, every call has only polynomial computation
  overhead and we proceed to bound the total number of recursive calls.

  To this end, let us bound the size of the recursion tree, starting by bounding its depth.
  As discussed above, the crucial point is that each considered
  selection matrix $M^{F_\psi(x_1,\ldots,x_n)}_{x_1,x_i}$ in Line~\ref{line:testmatrixentry}
  has at least two blocks, otherwise, the relation would have been
  decomposable, because by definition of blockwise decomposability every block of $M^{F_\psi(x_1,\ldots,x_n)}_{x_1, x_2}$ is decomposable. Therefore, the test for empty entries will succeed at
  least once and each considered variable domain shrinks. Therefore,
  in every root-leaf-path in the recursion tree, there are at most
  $|D|$ recursive calls in Line~\ref{line:recursebecauseshrinking}. Moreover, on such a path there cannot be two consecutive calls from Line~\ref{line:recursivebecausedecompose}, because we decompose into indecomposable factors before any such call. It follows that the recursion tree has depth at most $2|D|+1$.
  
  Let the height $h_v$ of a node $v$ in the recursion tree be the distance to the furthest leaf in the subtree below $v$. Let $n_v$ be the number of variables of the constraint in that call. We claim that the number of leaves below $v$ is at most $|D|^{h_v} n_v$. We show this by induction on $h_v$. If $h_v=0$, then $v$ is a leaf, so we make no further recursive calls. This only happens if $n_v=1$ and the claim is true. Now consider $h_v> 0$. Let $v_1, \ldots, v_r$ be the children of~$v$. If in $v$ we make a recursive call as in Line~\ref{line:recursivebecausedecompose}, then $n_{v_1}+ \ldots + n_{v_r} = n_v$. Also for all $v_i$ we have $h_{v_i}< h_v$ and the number of leaves below $v$ is bounded by $\sum_{i\in [r]} n_i |D|^{h_{v_i}} \le d^{h_v} \sum_{i\in [r]} n_i = |D|^{h_v} n$. If in $v$ we make a recursive call in Line~\ref{line:recursebecauseshrinking}, then we know that $r\le |D|$, because we make at most $|D|$ calls. Moreover, we have again that $h_{v_i}< h_v$, so the number of leaves below $v$ is bounded by $\sum_{i\in [r]} n_{v_i} |D|^{v_i}\le n_v |D|^{h_v}$ which completes the induction and thus proves the claim.
  
  It follows that the recursion tree of the algorithm has at most $n |D|^{2|D|+1}$ leaves and thus at most $2n |D|^{2|D|+1}$ nodes. Since we add at most one FDD-node per recursive call, this is also a bound for the size of the computed FDD.   
\end{proof}

\section{Lower Bounds}\label{sct:lower}

In this section, we will prove the lower bounds of Theorem~\ref{thm:ODD} and Theorem~\ref{thm:FDD}. In the proofs, we will use the approach developed in~\cite{BovaCMS16} that makes a connection between DNNF size and so-called rectangle covers. So let us start the section with some definitions.

Let $\vec u$ be a variable vector, $\tilde u =\tilde x \cup \tilde y$ and $\tilde x \cap \tilde y = \emptyset$. Then we say that the constraint  $\mathfrak R(\vec u)$  is a \emph{(combinatorial) rectangle} with respect to the partition $(\tilde x, \tilde y)$ if and only if $\mathfrak R(\vec u) = \pi_{\tilde x}(\mathfrak R(\vec u)) \times \pi_{\tilde y}(\mathfrak R(\vec u))$\footnote{While rectangles are constraints, we use a special typesetting for them, writing them in fraktur, to stress their decomposability property.}. 
Let $Z\subseteq \tilde u$ and $\beta>0$. Then we call the partition $(\tilde x, \tilde y)$ \emph{$Z$-$\beta$-balanced} if $\frac{\beta |Z|}{2} \le |\tilde x \cap Z| \le \beta |Z|$. A rectangle $\mathfrak R(\vec u)$ is called a \emph{$Z$-$\beta$-balanced rectangle} if it is a rectangle with respect to a $Z$-$\beta$-balanced partition.
A \emph{$Z$-$\beta$-balanced rectangle cover} $\mathcal R$ of a constraint $R(\vec u)$ is defined to be a set of $Z$-$\beta$-balanced rectangles such that $\sol(R(\vec u)) = \bigcup_{\mathfrak R(\vec u)\in \mathcal R} \sol(\mathfrak R(\vec u))$. The size of $\mathcal R$ is the number of rectangles in it.

We will use the following variant of the main result of~\cite{BovaCMS16} which gives a connection between the size of rectangle covers and DNNFs.

\begin{lemma}\label{lem:DNNFmakeRectangles}
        Let $O$ be a \textup{DNNF} of size $s$ representing a constraint $R(\vec x)$ and let $Z\subseteq \tilde x$. Then, for every $\beta>0$, there is a $Z$-$\beta$-balanced rectangle cover of $R(\vec x)$ of size $s$. Moreover, if~$O$ is structured, then the rectangles in the cover are all with respect to the same variable partition.
\end{lemma}
The proof of Lemma~\ref{lem:DNNFmakeRectangles} is very similar to
existing proofs in~\cite{BovaCMS16}, so we defer it to Appendix~\ref{app:DNNFmakeRectangles}.
\begin{toappendix}
\section{Proof of Lemma~\ref{lem:DNNFmakeRectangles}}\label{app:DNNFmakeRectangles}
In the proof of Lemma~\ref{lem:DNNFmakeRectangles}, we will again use the concept of proof trees, see Section~\ref{sct:preliminaries}.
The idea of the proof of Lemma~\ref{lem:DNNFmakeRectangles} is to partition $\mathcal T(O)$ in the representation (\ref{eq:prooftrees}), guided by the circuit $O$. To this end, we introduce some more notation. Let, for every gate $v$,  $\mathcal T(v,O)$ denote the set of proof trees of $O$ that contain $v$. Moreover, let $\var(O)$ denote the variables appearing in $O$ and $\var(v,O)$ the variables that appear in $O$ below the gate $v$. Finally, let $\mathfrak S(v,O):=\bigcup_{T\in \mathcal T(v,O)} S(T)$.
\begin{claim}\label{clm:rectangle}
        $\mathfrak S(v,O)$ is a rectangle w.r.t.~$(\var(v,O), \var(O)\setminus \var(v,O))$.
\end{claim}
\begin{proof}
        Every proof tree in $\mathcal T(v,O)$ can be partitioned into a part below $v$ and the rest. Moreover, any such proof trees $T_1, T_2$ can be combined by combining the part below $v$ from $T_1$ and the rest from $T_2$. The claim follows directly from this observation.
\end{proof}
With Claim~\ref{clm:rectangle}, we only have to choose the right gates of $O$ to partition $\mathcal T(O)$ to prove Lemma~\ref{lem:DNNFmakeRectangles}. To this end, we construct an initially empty rectangle cover $\mathcal R$ iteratively: while~$O$ still captures an assignment $\alpha$, choose a proof tree $T$ capturing $\alpha$ (which is guaranteed to exist by (\ref{eq:prooftrees})). By descending in $T$, choose a gate $v$ such that $\var(v,O)$ a fraction of the variables in $Z$ between $\frac{\beta}{2}$ and $\beta$\footnote{There is a small edge case here in which $O$ does not contain enough of the variables in $Z$. In that case, we simply take $S(O)$ as a rectangle, balancing it by adding the non-appearing variables appropriately.}. Add $\mathfrak S(v,O)$ to $\mathcal R$, delete $v$ from $O$ and repeat the process. When the iteration stops, we have by Claim~\ref{clm:rectangle} and the choice of $v$ constructed a set $\mathcal R$ of $Z$-balanced rectangle covers. Moreover, $\bigcup_{\mathfrak R\in \mathcal R} S(\mathfrak R) \subseteq S(O)$ by construction. Finally, since in the end $O$ does not capture any assignments anymore, every assignment $\alpha$ initially captured must have been computed by one of the proof trees of $O$ that got destroyed by deleting one of its gates $v$. Thus $\alpha \in \mathfrak S(v,O)\in \mathcal R$ and we have
\begin{align*}
        S(O)= \bigcup_{\mathfrak R(\vec x)\in \mathcal R}S(\mathfrak R(\vec x)),
\end{align*}
which shows the claim of Lemma~\ref{lem:DNNFmakeRectangles} for the unstructured case.

If $O$ is structured, we choose the vertices $v$ in the iteration slightly differently. Let $\mathcal T$ be the v-tree of $O$. Then we can choose a vertex $v^*$ in $\mathcal T$ that has between $\frac{\beta}{2}$ and $\beta$ of the variables in $Z$ as labels on leaves below $v^*$. Let $X_{\downarrow}$ be the variable in $\mathcal T$ below $v$ and let $X_{\uparrow}$ be the rest of the variables. Now in the construction of $\mathcal R$, in the proof tree $T$ we can find~a gate $v$ below which there are only variables in $X_{\downarrow}$ and which is closest to the root with this property. Then, by Claim~\ref{clm:rectangle}, $\mathfrak S(v, O)$ is a rectangle w.r.t.~$(X_{\downarrow}, X_{\uparrow})$ and thus in particular $Z$-balanced. Since all rectangles we choose this way are with respect to the same partition $(X_{\downarrow}, X_{\uparrow})$ which depends only on $\mathcal T$, the rest of the proof follows as for the unstructured case.
\end{toappendix}

\subsection{Lower Bound for DNNF}\label{sec:DNNFlower}

In this section, we show the lower bound for Theorem~\ref{thm:FDD} which we reformulate here.

\begin{proposition}\label{prop:lowerFDD}
Let $\Gamma$ be a constraint language that is not strongly blockwise decomposable. Then there is a family of $\Gamma$-formulas $\varphi_{n}$ of size $\Theta(n)$ and
$\varepsilon>0$ such that any \textup{DNNF} for $F_{\varphi_n}$ has
size at least $2^{\varepsilon\left\Vert F_{\varphi_n}\right\Vert }$.
\end{proposition}

In the remainder of this section, we show Proposition~\ref{prop:lowerFDD}, splitting the proof into two cases. 
First, we consider the case where $M_{x,y}^{R(x,y,\vec z)}$ is not a proper block matrix. 

\begin{lemma}\label{lem:NoBlockNoDNNF}
Let $R(x,y,\vec z)$ be a constraint such that $M_{x,y}^{R(x,y,\vec z)}$ is not a proper block matrix. Then 
there is a family of $\{R\}$-formulas $\varphi_{n}$ and
$\varepsilon>0$ such that any \textup{DNNF} for $F_{\varphi_n}$ has size at least $2^{\varepsilon\left\Vert \varphi_n\right\Vert }$.
\end{lemma}

\newcommand{\degree}{\ensuremath{r}}

In the proof of Lemma~\ref{lem:NoBlockNoDNNF}, we will use a specific family of graphs. We remind the reader that a matching is a set of edges in a graph that do not share any end-points. The matching is called induced if the graph induced by the end-points of the matching contains exactly the edges of the matching.
\begin{lemma}\label{lem:expanders}
        There is an integer $\degree$ and constants $\alpha>0$, $\varepsilon> 0$ such that there is an infinite family $(G_n)$
         of  bipartite graphs with maximum degree at most $\degree$ such that for each set $X\subseteq V(G_{n})$ with $|X|\le \alpha n=\alpha |V(G_n)|/2$ there is an induced matching of size at least $\varepsilon|X|$ in which each edge has exactly one endpoint in $X$.
\end{lemma}
The proof of Lemma~\ref{lem:expanders} uses a specific class of so-called expander graphs. Since the arguments are rather standard for the area, we defer the proof to Appendix~\ref{app:expander}.
\begin{toappendix}
        \section{Bipartite graphs with large induced matchings over every cut}\label{app:expander}

        In this appendix, we will show how to prove Lemma~\ref{lem:expanders}. The construction is based on expander graphs in a rather standard way, but since we have not found an explicit reference, we give it here for completeness.

        We will use the following construction.

        \begin{lemma}\label{lem:probabilistic}
            There are constants $\degree\in \mathbb{N}$, $\alpha >0$,
            $c > 1$ such that there is a class of bipartite graphs
            $G_n= (A_n, B_n, E_n)$ of degree at most $\degree$ and with $|A_n|= |B_n| = n$ such that for every set $S\subseteq A_n$ or $S\subseteq B_n$ with $|S|\le \alpha n$ we have $|N(S)| > c |S|$.
        \end{lemma}        
        \begin{proof}
            The proof is an adaption of \cite[Theorem 5.6]{MotwaniR95} to our slightly different setting, using the probabilistic method. We will choose the constants $\degree,\alpha, c$ later to make the calculations work, so we let them be variable for now. We fix $n$ and construct $G_n$ as follows: Set $A_n:=\{(0, i)\mid i\in [n]\}$ and $B_n:=\{(1,i)\mid i\in [n]\}$. Then choose $\degree$ permutations $\pi_1, \ldots, \pi_\degree$ of $[n]$ and set $E_n:= \{((0,i), (1, \pi_j(i)))\mid i\in [n], j\in [\degree]\}$. Then $G_n$ is by construction bipartite and has maximum degree of $\degree$ as required, and it remains only to show the condition on the neighborhoods.
            
            Let $\mathcal E_s$ be the random event that there is a set
            $S\subset A_n$ of size $s$ with at most $c s$
            neighbors. There are $\binom{n}{s}$ possible choices of
            such a set $S$. Also, for every $S$ there are $\binom{n}{c
              s}$ sets $T$ in which the neighbors of $S$ can be in
            case $E_s$ is true for $S$. Since the probability of
            $N(S)\subseteq T$ only depends on the size of $T$ but not
            on $T$ itself, we get
            \begin{align*}
                \Pr(\mathcal E_s) &\le \binom{n}{s}\binom{n}{cs} \Pr(N(S)\subseteq \{(1, i)\mid i\in [cs]\})\\
                & \le \left(\frac{ne}{s}\right)^s\left(\frac{ne}{cs}\right)^{cs} \Pr(N(S)\subseteq \{(1, i)\mid i\in [cs]\}).
            \end{align*}
        We have that $N(S)\subseteq \{(1, i)\mid i\in [cs]\}$ if and only if the $\degree$ permutations $\pi_j$ all map $S$ into $[cs]$. So let us first bound the number of permutations which map $S$ into $[cs]$: we first choose the elements into which $S$ is mapped; there are $\prod_{i=0}^{s-1}(cs -i)$. Then we map the rest of $A$ arbitrarily; there are $(n-s)!$ ways of doing this. So the overall number of such permutations is $\prod_{i=0}^{s-1}(cs -i) (n-s)!$. Since the permutations $\pi_j$ are chosen independently, we get
        \begin{align*}
            \Pr(N(S)\subseteq \{(1, i)\mid i\in [cs]\}) &= \left( \frac{\prod_{i=0}^{s-1}(cs -i) (n-s)!}{n!}\right)^\degree\\
            &= \left(\prod_{i=0}^{s-1} \frac{cs -i}{n-i}\right)^\degree\\
            & \le \left(\frac{cs}{n-s}\right)^{\degree s}. 
        \end{align*}
        Plugging this in and then using $s\le \alpha n$, we get
        \begin{align*}
            \Pr(\mathcal E_s) &\le
                                \left(\frac{ne}{s}\right)^s\left(\frac{ne}{cs}\right)^{cs}\left(\frac{cs}{n-s}\right)^{\degree
                                s}\\
            &\le \left( \left(\frac{s}{n} \right)^{\degree-c-1}  e^{1+c}  c^{\degree-c}\left(\frac{1}{1-\alpha}\right)^\degree \right)^s \\
            &\le \left( \alpha^{\degree-c-1}  e^{1+c} c^{\degree-c} \left(\frac{1}{1-\alpha}\right)^\degree \right)^s.
        \end{align*}
        We now set our constants to $\alpha= 1/5$, $\degree=18$, and $c=1.1$ and get
        \begin{align*}
        \Pr(\mathcal E_s) \le 0.1^s.
        \end{align*}
        Now let $\mathcal E_A$ be the event that there is a subset $S$ of $A_n$ of size at most $\alpha n$ that has at most $c|S|$ neighbors. We get 
        \begin{align*}
            \Pr(\mathcal E_A) \le \sum_{s=1}^{\alpha n} \Pr(\mathcal E_s) \le \sum_{s=1}^{\alpha n} 0.1^s \le \frac{0.1}{1-0.1} \le 0.2.
        \end{align*}
        The same analysis for subsets of $B$ yields that the probability that there is a set $S$ in $A_n$ or $B_n$ of size at most $\alpha n$ that has too few neighbors is at most $0.4$. It follows that there is a graph $G_n$ with the desired properties.
        \end{proof}

        We call a graph $G$ an $(n,\degree,c,\alpha)$ vertex expander
        if $|V(G)|=n$, the maximum degree is at most $\degree$, and for all sets $S$ of at most $\alpha n$ vertices, the neighborhood $N(S)=\{v\in V(G)\setminus S \mid \exists u \in S, uv\in E\}$ has size at least $c |S|$.
        
        \begin{corollary}\label{cor:app1}
                There are constants $\degree\in \mathbb{N}$, $\alpha >0$, $c' > 0$ such that there is a class of bipartite $(n,\degree,c', \alpha)$-expander with $n$ vertices in every color class for infinitely many values $n$.
        \end{corollary}
        \begin{proof}
            We take the graphs $G_n$ from Lemma~\ref{lem:probabilistic} with the same values $\alpha$, $\degree$ and $c':= \frac{c-1}{2}$. Fix $G_n$ and consider any set $S$ of size at most $\alpha n$. Assume w.l.o.g.~that $|S\cap A_n| \ge |S\cap B_n|$, so $|A_n\cap S| \ge |S|/2$. Then we get that 
            \begin{align*}
                |N(S)\setminus S| &\ge |N(S\cap A_n)\setminus S\cap B|\\
                & \ge |N(S\cap A_n) | - |S\cap B|\\
                &\ge c |S\cap A_n| - |S|/2\\
                & =\frac{c-1}{2} |S|\\ 
                &= c'|S|.
            \end{align*}
        \end{proof}
        The class of graphs in Corollary~\ref{cor:app1} will be the class of graphs for Lemma~\ref{lem:expanders}. We now construct the distant matchings. To this end, consider a graph $G=(V,E)$ from this class and a set $S$ of vertices of size at most $\alpha n$. First construct a matching between $S$ and $V\setminus S$. Since $S$ has $c' |S|$ neighbors outside of $S$ and every vertex had degree at most $\degree$, one can greedily find such a matching of size $\frac{c'|S|}{\degree}$. In a second step, we choose an induced matching out of this matching greedily. Since every edge has at most $O(\degree^2)$ edges at distance $2$, this yields an induced matching of size $\Omega(|V|/\degree^3)$ which completes the proof.
\end{toappendix}
 
\begin{proof}[Proof of Lemma~\ref{lem:NoBlockNoDNNF}]

If $M_{x,y}^{R(x,y,\vec z)}$ is not a proper block matrix, then, by Lemma~\ref{lem: 2x2 submatrix}, the matrix $M_{x,y}^{R(x,y,\vec z)}$ has a $2\times2$-submatrix 
with exactly three non-empty entries. So let $a,b,c,d \in D$ such that $M_{x,y}^{R(x,y,\vec z)}(b,d)=\emptyset$ and 
$M_{x,y}^{R(x,y,\vec z)}(a,c)$, $M_{x,y}^{R(x,y,\vec z)}(a,d)$ and $M_{x,y}^{R(x,y,\vec z)}(b,c)$ are all non-empty.

We describe a construction that to every bipartite graph $G= (A,B,E)$
gives a formula $\varphi_G$ as follows: for every vertex $u\in A$, we
introduce a variable $x_u$ and for every vertex $v\in B$ we introduce
a variable $x_v$. Then, for every edge $e=uv\in E$ where $u\in A$ and
$v\in B$,
we add a constraint $R(x_u, x_v, \vec z_e)$ where $\vec z_e$
consists of variables only used in this constraint.

Let $(\varphi_{n})$ be the family of formulas defined by
$\varphi_{n}=\varphi_{G_{n}}$ where $(G_{n})$ is the family from
Lemma~\ref{lem:expanders}. Clearly, $\|\varphi_n\|=
\Theta(|E(G_n)|)=\Theta(n)$, as required. Fix~$n$ for the remainder of
the proof and consider
$G_n =(A,B,E)$.
We fix the notation $X_A:= \{x_s\mid s\in A\}$, $X_B:= \{x_t\mid t\in B\}$ and $X:= X_A\cup X_B$.
Let $\varphi'_n$ 
be the formula we get from $\varphi_n$ by restricting all variables $x_s\in X_A$ to $\{a,b\}$ and all variables $x_t\in X_B$ to $\{c,d\}$ by adding some unary constraints.
Let $\mathcal{R}$ be an $X$-$\alpha$-balanced rectangle cover of $\varphi'_n$ where $\alpha$ is the constant from Lemma~\ref{lem:expanders}. We claim that the size of $\mathcal{R}$ is at least $2^{\varepsilon' n}$, where
\[
\varepsilon' = \frac{1}{2} \cdot \varepsilon \cdot \alpha \cdot \log_2\left( 1 + \frac{1}{2^{\degree+1}|R|^{\degree^2}} \right)
\]
and $\degree$ is the degree of $G_n$.
To prove this, we first show that for every $\mathfrak R\in \mathcal{R}$, 
$
2^{\varepsilon' |X|} \cdot |\sol(\mathfrak R)|  \le |\sol(F_{\varphi'_n})|.
$
So let $\mathfrak R(\vec v, \vec w)= \mathfrak R_1(\vec v)\times
\mathfrak R_2(\vec w)\in \mathcal{R}$. Since $\mathfrak R$ is an
$X$-$\alpha$-balanced rectangle, we may assume
$\alpha |X|/2 \leq | \tilde v \cap X | \leq \alpha |X|$.
By choice of~$G_n$, we have that there is an induced matching $\mathcal M$ in $G_n$ of size at least $\varepsilon \alpha |X|/2$ consisting of edges that have 
one endpoint corresponding to a variable in $\tilde v$ and one
endpoint corresponding to a variable in $\tilde w$. Consider an edge
$e=st\in \mathcal M$. Assume that $x_s\in \tilde v$ and $x_t\in \tilde
w$.
Since we have $\mathfrak R(\vec v, \vec w)= \mathfrak R_1(\vec v)\times \mathfrak R_2(\vec w)$, we get 
\begin{align*}
\pi_{\{x_s, x_t\}} \mathfrak R(\vec v, \vec w) &= \pi_{\tilde v \cap \{x_s, x_t\}}(\mathfrak R_1(\vec v)) \times \pi_{\tilde w \cap \{x_s, x_t\}}(\mathfrak R_2(\vec w))\\
  &= \pi_{\{x_u\}}(\mathfrak R_1(\vec v)) \times \pi_{\{x_t\}}(\mathfrak R_2(\vec w)).
\end{align*}

By construction $\pi_{\{x_s, x_t\}} \mathfrak R(\vec v, \vec w)
\subseteq \{(a,c),(a,d),(b,c)\}$, so it follows that either
$\pi_{\{x_s, x_t\}} \mathfrak R(\vec v, \vec w) \subseteq
\{(a,c),(a,d)\} = \{a\} \times \{c,d\}$ or $\pi_{\{x_s, x_t\}}
\mathfrak R(\vec v, \vec w) \subseteq \{(a,c),(b,c)\} = \{a,b\} \times
\{c\}$. Assume w.l.o.g.~that $(b,c) \notin \pi_{\{x_s, x_t\}} \mathfrak R(\vec
v, \vec w)$ 
(the other case can be treated analogously). It follows
that for each solution $\beta\in \sol(\mathfrak R(\vec v, \vec w))$, we get a solution
$q(\beta)\in \sol(F_{\varphi_{n}'})\backslash \sol(\mathfrak R(\vec v, \vec w))$
by setting
\begin{itemize}
\item $q(\beta)(x_{s}):=b$,
\item For all $x_{\ell}\in N(x_{s})$, we set $q(\beta)(x_\ell) := c$,
\item For all $x_{\ell}\in N(x_{s})$ and all $x_{m}\in N(y_{\ell})$ we set $q(\beta)(\vec{z}_{m\ell}) := \vec g$ where $\vec g$ is such that
$(b,c,\vec g) \in R(x,y,\vec z)$.
\end{itemize}
Note that values $\vec g$ exist because $M_{x,y}^{R(x,y,\vec z)}(b,c)$ is non-empty. Observe that 
for two different solutions $\beta$ and $\beta'$ the solutions $q(\beta)$
and $q(\beta')$ may be the same. However, we can bound the number $\left|q^{-1}(q(\beta))\right|$,
giving a lower bound on the set of solutions not in $\mathfrak R(\vec v, \vec w)$.
To this end, suppose that $q(\beta)=q(\beta')$. Since $q$ only changes the values
of $x_{s}$, exactly $\degree$ $x_\ell$-variables and at most $\degree^2$ vectors of $\vec{z}$-variables (the two latter bounds come from the degree bounds on $G_n$),
$q(\beta)=q(\beta')$ implies that $\beta$ and $\beta'$ coincide
on all other variables. It follows that
\[
\left|q^{-1}(q(\beta))\right|\le 2^{\degree+1}|R|^{\degree^2},
\]
because there are only that many possibilities for the variables that
$q$ might change. By considering $\{x_{s},x_{t}\}$, we have shown that
\[
|\sol(F_{\varphi_n'})| \ge |\sol(\mathfrak R(\vec v, \vec w))| + \frac{1}{2^{\degree+1}|R|^{\degree^2}} |\sol(\mathfrak R(\vec v, \vec w))|.
\]
So we have constructed $\frac{|\sol(\mathfrak R(\vec v, \vec w))|}{2^{\degree+1}|R|^{\degree^2}}$
solutions not in $\mathfrak R(\vec v, \vec w)$. Now we consider not only one
edge but all possible subsets $I$ of edges in $\mathcal M$: 
for a solution $\beta\in \sol(\mathfrak R(\vec v, \vec w))$, the assignment $q_{I}(\beta)$ is constructed as the $q$
above, but for all edges $e\in I$. Reasoning as above, we get
\[
\left|q_{I}^{-1}(q_{I}(\beta))\right|\le\left(2^{\degree+1}|R|^{\degree^2}\right)^{|I|}.
\]
It is immediate to see that $q_{I}(\beta)\neq q_{I'}(\beta)$
for $I\neq I'$. Thus we get 
\begin{align*}
|\sol(F_{\varphi_n'})| & \ge\sum_{m=0}^{\frac{1}{2}\varepsilon \alpha |X|}\binom{\frac{1}{2}\varepsilon \alpha |X|}{m} \left(\frac{1}{2^{\degree+1}|R|^{\degree^2}}\right)^{m} |\sol(\mathfrak R)|\\
 & =\left(1+\frac{1}{2^{\degree+1}|R|^{\degree^2}}\right)^{\frac{1}{2}\varepsilon \alpha |X|} |\sol(\mathfrak R)|%
  =2^{\varepsilon' n} |\sol(\mathfrak R)|.
\end{align*}
It follows that every $X$-$\alpha$-balanced rectangle cover of $F_{\varphi_n'}$ has to have a size of at least $2^{\varepsilon' |X|}$. 
With Lemma \ref{lem:DNNFmakeRectangles} and Lemma~\ref{lem:restrictDNNF} it follows that any DNNF for $F_{\varphi_n}$ has to have a size of at least $2^{\varepsilon' n}$. 
\end{proof}

Now we consider the case that $M_{x,y}^{R(x,y,\vec z)}$ is a proper block matrix, but $R(x,y,\vec z)$ is not blockwise decomposable in some pair of variables $x, y$.

\begin{lemma}\label{lem:NoDecNoDNNF}
Let $R(x,y,\vec z)$ be a relation such that $M_{x,y}^{R(x,y,\vec z)}$ is a proper block matrix but $R(x,y,\vec z)$ in not blockwise decomposable in $x$ and $y$. Then there is a family of formulas $\varphi_{n}$ and
$\varepsilon>0$ such that a \textup{DNNF} for $F_{\varphi_n}$ needs to have a
size of at least $2^{\varepsilon\| \varphi_n\|}$.
\end{lemma}
\begin{proof}
The proof follows the same ideas as that of Lemma~\ref{lem:NoBlockNoDNNF}, so we state only the differences. If $M_{x,y}^{R(x,y,\vec z)}$ is a 
 proper block matrix but $R(x,y,\vec z)$ is not blockwise decomposable in $x$ and $y$, there is a block $B = D_x \times D_y$ such that $R(x, y, \vec z)|_{(x,y 
 )\in B}$ is not decomposable, in such a way that $x$ and $y$ appear in different factors. This implies that if $\vec z = \vec v \dot \cup \vec 
  w$, $V \subseteq \pi_{x,\vec v} R(x,y,\vec z)|_{(x,y)\in B}$ and $W \subseteq \pi_{y,\vec w} R(x,y,\vec z)|_{(x,y)\in B}$, then $V\times W \subsetneq R|_{(x,y)\in B}$. 
  For one
matching edge  $e=uv \in \mathcal M$
  the projection onto $\mathfrak R$ is
  of the form $V\times W$, so there must exist $(a,b,\vec c) \in R$
  and not in $\pi_{\{x_u, x_v, z_e\}} \mathfrak R$.
  This $(a,b,\vec c)$ is used for the edge $e$ to construct solutions in $sol(F_n) \setminus \mathfrak R$. Since $B$ is a block, we define $
 q(\beta)$ as follows:
\begin{itemize}
\item[(a)] $q(\beta)(x_{u}):=a$, $q(\beta)(x_v):=b$
\item[(b)] For all $x_{\ell}\in N(x_{v})$ we set  $q(\beta)(\vec{z}_{\ell v}) := \vec g$ such that $(\beta(x_{\ell}),b,\vec g) \in R$.
\item[(c)] For all $x_{\ell}\in N(x_{u})$ we set
  $q(\beta)(\vec{z}_{u\ell}) := \vec g$ such that
  $(a,\beta(x_{\ell}),\vec g)\in R$.
\end{itemize}
To bound $|q^{-1}_I(q_I(\beta))|$, note that we have at most $|D_x||D_y|$
possibilities in case (a) and, since $|N(x_u)|,|N(x_v)|\leq \degree$, at most $|R|^\degree$ possibilities in case (b)
and (c), respectively.
It follows that we can bound the number $|q^{-1}_I(q_I(\beta))|$ for a solution $\beta$ by
\[
|q^{-1}_I(q_I(\beta))|\le \left(|D_x||D_y||R|^{2\degree}\right)^{|I|}
\]
so we get:
\[
\varepsilon = \frac{1}{3} \cdot \varepsilon \cdot \alpha \cdot \log_2\left( 1 + \frac{1}{|D_x||D_y||R|^{2\degree}} \right).
\]
The rest of the proof is unchanged.
\end{proof}

We now have everything in place to prove Proposition~\ref{prop:lowerFDD}.

\begin{proof}[Proof of Proposition~\ref{prop:lowerFDD}]
Since $\Gamma$ is not strongly blockwise decomposable, there is a
relation $R\in \langle\Gamma\rangle$ that is not blockwise
decomposable in $x$ and $y$. Then, by definition of co-clones and
Lemma~\ref{lem:closedProjectionSeletion}, there is a $\Gamma$-formula
that defines $R(\vec u)$. 
If there are variables $x,y\in \tilde u$ such that $M_{x,y}^{R(\vec u)}$ is not a proper block matrix, then we can apply Lemma~\ref{lem:NoBlockNoDNNF} to get $\{R\}$-formulas that require exponential size DNNF. Then by substituting all occurrences of $R$ in these formula by the $\Gamma$-formula defining $R$, we get the required hard $\Gamma$-formulas. If all $M_{x,y}^{R(\vec u)}$ are proper block matrices, then there are variables $x, y$ such that $M_{x,y}^{R(\vec u)}$ is not blockwise decomposable. Using Lemma~\ref{lem:NoDecNoDNNF} and reasoning as before, then completes the proof.
\end{proof}

\subsection{Lower Bound for structured DNNF}\label{sec:SDNNFlower}

In this section, we prove the lower bound of Theorem~\ref{thm:ODD} which we formulate here again. 

\begin{proposition}\label{prop:lowerODD}
        Let $\Gamma$ be a constraint language that is not strongly uniformly blockwise decomposable. Then there is a family of $\Gamma$-formulas $\varphi_{n}$ of size $\Theta(n)$ and
        $\varepsilon>0$ such that any structured \textup{DNNF} for $F_{\varphi_n}$ has
        size at least $2^{\varepsilon\left\Vert \varphi_{n}\right\Vert }$.
\end{proposition}

Note that for all constraint languages that are not strongly blockwise decomposable, the result follows directly from Proposition~\ref{prop:lowerFDD}, so we only have to consider constraint languages which are strongly blockwise decomposable but not strongly uniformly blockwise decomposable. We start with a simple observation.

\begin{observation}\label{obs:projectRectangles}
                Let $\mathfrak R(\vec x, \vec y)$ be a rectangle with respect to the partition $(\tilde x, \tilde y)$. Let $Z\subseteq \tilde x\cup \tilde y$, then $\pi_Z(\mathfrak R(\tilde x, \tilde y))$ is a rectangle with respect to the partition $(\tilde x\cap Z, \tilde y\cap Z)$.
\end{observation}

We start our proof of Proposition~\ref{prop:lowerODD} by considering a special case.

\begin{lemma}\label{lem:lowerboundmatchingformula}
                Let $C = R(x, y, \vec{z})$ be a constraint such that there are two assignments $\alpha, \beta\in \sol(C)$ such that for every partition $\tilde z_1,\tilde z_2$ of $\tilde z$ we have that $\alpha|_{\{x\}\cup \tilde z_1}\cup \beta|_{\{y\}\cup \tilde z_2}\notin \sol(C)$ or $\alpha|_{\{y\}\cup \tilde z_2}\cup \beta|_{\{x\}\cup \tilde z_1}\notin \sol(C)$. Consider 
                \begin{align*}
                                F_n := \bigwedge_{i\in [n]}  R(x_i, y_i, \vec z_i)
                \end{align*}
                where the $\vec z_i$ are disjoint variable vectors.
                Let $(\tilde x,\tilde y)$ be a variable partition of the variables of $\varphi_n$ and $\mathcal R$ be a rectangle cover of $F_n$ such that each rectangle $\mathfrak R$ in $\mathcal R$ respects the partition $(\tilde x,\tilde y)$. If for all $i\in [n]$ we have that all $x_i\in \tilde x$ and $y_i\in \tilde y$ or $x_i\in \tilde y$ and $y_i\in \tilde x$, then $\mathcal R$ has size at least $2^n$.
\end{lemma}
\begin{proof}
                We use the so-called fooling set method from
                communication complexity, see
                e.g.~\cite[Section~1.3]{KushilevitzN97}. To this end,
                we will construct a set $\mathcal S$ of satisfying
                assignments of $F_n$ such that every rectangle of $\mathcal R$ can contain at most one assignment in $\mathcal S$.
                
                So let $\alpha_i$ be the assignment to $\{x_i, y_i\}\cup \tilde z_i$ that assigns the variables analogously to $\alpha$, so $\alpha_i(x_i):= \alpha(x)$, $\alpha_i(y_i):= \alpha(y)$, and $\alpha_i(\vec z_i):= \alpha(\vec z)$. Define analogously $\beta_i$. Then the set $\mathcal S$ consists of all assignments that we get by choosing for every $i\in [n]$ an assignment $\delta_i$ as either $\alpha_i$ or $\beta_i$ and then combining the $\delta_i$ to one assignment to all variables of $F_n$. Note that $\mathcal S$ contains $2^n$ assignments and that all of them satisfy $F_n$, so all of them must be in a rectangle of $\mathcal R$.
                
                We claim that none of the rectangles $\mathfrak R$ of $\mathcal R$ can contain more than one element $\delta \in \mathcal S$. By way of contradiction, assume this were not true. Then there is an $\mathfrak R$ that contains two assignments $\delta, \delta'\in \mathcal S$, so there is an $i\in [n]$ such that in the construction of $\delta$ we have chosen $\alpha_i$ while in the construction of $\delta'$ we have chosen $\beta_i$. Let $\tilde x_i := \tilde z_i \cap \tilde x$ and $\tilde y_i:= \tilde z_i \cap \tilde y$. Since $\delta, \delta'\in \sol(\mathfrak R)$, we have that $\alpha_i, \beta_i\in \sol(\pi_{\{x_i, y_i\}\cup \tilde x_i\cup \tilde y_i}(\mathfrak R))$. Moreover, by Observation~\ref{obs:projectRectangles}, $\pi_{\{x_i, y_i\}\cup \tilde x_i\cup \tilde y_i}(\mathfrak R)$ is a rectangle and so we have that $\alpha_i|_{\{x_i\}\cup \tilde x_i}\cup  \beta_i|_{\{y_i\}\cup y_i}\in \sol(\pi_{\{x_i, y_i\}\cup \tilde x_i\cup \tilde y_i}(\mathfrak R))$ and $\alpha_i|_{\{y_i\}\cup y_i}\cup  \beta_i|_{\{x_i\}\cup y_i}\in \sol(\pi_{\{x_i, y_i\}\cup \tilde x_i\cup \tilde y_i}(\mathfrak R))$. But $\mathfrak R$ consists only of solutions of $F_n$ and thus $\sol(\pi_{\{x_i, y_i\}\cup \tilde x_i\cup \tilde y_i}(\mathfrak R))\subseteq \sol(R(x_i, y_i, \vec z_i))$, so $\alpha_i|_{\{x_i\}\cup \tilde x_i}\cup  \beta_i|_{\{y_i\}\cup y_i}, \alpha_i|_{\{y_i\}\cup y_i}\cup  \beta_i|_{\{x_i\}\cup y_j}\in \sol(R(x_i, y_i, \vec z_i))$. It follows by construction that there is a partition $(\tilde x, \tilde y)$ of $\tilde z$ such that $\alpha|_{\{x\}\cup \tilde x}\cup \beta|_{\{y\}\cup \tilde y}$ and $\alpha|_{\{y\}\cup \tilde y}\cup \beta|_{\{x\}\cup \tilde x}$ are in $\sol(C)$. This contradicts the assumption on $\alpha$ and $\beta$ and thus $\mathfrak R$ can only contain one assignment from $\mathcal S$.
                
                Since $\mathcal S$ has size $2^n$ and all of assignments in $\mathcal S$ must be in one rectangle of~$\mathcal R$, it follows that $\mathcal R$ consists of at least $2^n$ rectangles.
\end{proof}

We now prove the lower bound of Proposition~\ref{prop:lowerODD}.

\begin{proof}[Proof of Proposition~\ref{prop:lowerODD}]
Since $\Gamma$ is not strongly uniformly blockwise decomposable, let
$R(x,y,\vec z)$ be a constraint that is pp-definable over $\Gamma$ and not uniformly blockwise decomposable in $x$ and $y$.
If $R(x,y, \vec z)$ is such that $M^{R(x,y, \vec z)}_{x,y}$ is not a proper block matrix, then the lemma follows directly from Lemma~\ref{lem:NoBlockNoDNNF}, so we assume in the remainder that $M^{R(x,y, \vec z)}_{x,y}$ is a proper block matrix. We denote for every block $D_1\times D_2$ of $R(x, y, \vec z)$ by $R^{D_1\times D_2}(x, y , \vec z)$ the sub-constraint of $R(x,y,\vec z)$ we get by restricting $x$ to $D_1$ and $y$ to $D_2$. Since $R(x, y, \vec z)$ is not uniformly blockwise decomposable, for every partition $(\tilde z_1, \tilde z_2)$ of $\tilde z$ there is a block $D_1\times D_2$ such that $R^{D_1\times D_2}(x, y , \vec z) \ne \pi_{\{x\}\cup \tilde z_1}R^{D_1\times D_2}(x, y, \vec z) \times \pi_{\{y\}\cup \tilde z_2}R^{D_1\times D_2}(x, y, \vec z)$.
                
Given a bipartite graph $G=(A,B,E)$, we construct the same formula $\varphi_G$ as in the proof of Lemma~\ref{lem:NoBlockNoDNNF}. Consider again the graphs $G_n$ of the family from Lemma~\ref{lem:expanders} and let $\varphi_n:=\varphi_{G_n}$. Fix $n$ in the remainder of the proof. Let $O$ be a structured DNNF representing $F_{\varphi_n}$ of size $s$. Then, by Proposition~\ref{lem:DNNFmakeRectangles}, there is  an $\alpha$-balanced partition $(X_1, X_2)$ of $X_A\cup X_B$ such that there is a rectangle cover of $F_{\varphi_n}$ of size at most $s$ and such that all rectangles respect the partition  $(X_1, X_2)$. Let $E'$ be the set of edges $uv\in E$ such that $x_u$ and $x_v$ are in different parts of the partition $(X_1, X_2)$. By the properties of $G_n$, there is an induced matching $\mathcal M$ of size $\Omega(|A|+|B|)$ consisting of edges in $E'$. 
                
For every edge $e=uv\in \mathcal M$ let $\tilde z_{e,1} := \tilde z_e\cap X_1$ and $\tilde z_{e,2} := \tilde z_2\cap X_2$. Assume that $x_u\in X_1$ and $y_v\in X_2$ (the other case is treated analogously). Then we know that there is a block $D_1\times D_2$ of $M^{R(x,y, \vec z)}_{x,y}$ such that 
\begin{align}\label{eq:doesnotdecompose}
R^{D_1\times D_2}(x_u, y_v , \vec z_e) \ne \pi_{\{x_u\}\cup \tilde z_{e,1}}R^{D_1\times D_2}(x_u, y_v , \vec z_e) \times \pi_{\{y_v\}\cup \tilde z_{e,2}}R^{D_1\times D_2}(x_u, y_v , \vec z_e).
\end{align} 
Since there are only at most $|D|$ blocks in $M^{R(x,y, \vec z)}_{x,y}$, there is a block $D_1\times D_2$ such that for at least $\Omega\left(\frac{|A|+|B|}{|D|}\right)= \Omega({|A|+|B|})$ edges Equation (\ref{eq:doesnotdecompose}) is true. Call this set of edges $E^*$.
                
Let $X^*:= \{x_u\mid u \text{ is an endpoint of an edge } e\in E^*\}$. We construct a structured DNNF $O'$ from $O$ by existentially quantifying all variables not in a constraint $R(x_u, y_v, \vec z_e)$ for $e=uv\in E^*$ and for all $x_u\in X^*$ restricting the domain to $D_1$ $u\in A$ and to $D_2$ if $u\in B$. Note that every assignment to $X^*$ that assigns every variable $x_v$ with $v\in A$ to a value in $D_1$ and every $x_v$ with $v\in B$ to a value in $D_2$ can be extended to a satisfying assignment of $F_{\varphi_n}$, because $D_1\times D_2$ is a block. Thus, $O'$ is a representation of 
\begin{align*}
F^*:= \bigwedge_{e=uv\in E^*} R^{D_1\times D_2}(x_u, y_v, \vec z_e).
\end{align*}

We now use the following simple observation.
\begin{claim}\label{clm:noproduct}
Let $R'(\vec x, \vec y)$ be a constraint such that $R'(\vec x, \vec y) \ne \pi_{\tilde x}(R'(\vec x, \vec y)) \times \pi_{\tilde y}(R'(\vec x, \vec y))$. Then there are assignments $a,b\in \sol(R'(\vec x, \vec y))$ such that $a|_{\tilde x} \cup b|_{\tilde y}\notin \sol(R'(\vec x, \vec y))$ or $a|_{\tilde x} \cup b|_{\tilde y}\notin \sol(R'(\vec x, \vec y))$.
\end{claim}
\begin{proof}
Since $R(\vec x, \vec y) \ne \pi_{\tilde x}(R(\vec x, \vec y)) \times \pi_{\tilde y}(R(\vec x, \vec y))$ and thus $R(\vec x, \vec y) \subsetneq \pi_{\tilde x}(R(\vec x, \vec y)^) \times \pi_{\tilde y}(R(\vec x, \vec y))$, we have that there is $\alpha\in \sol(\pi_{\tilde x}(R(\vec x, \vec y)))$ and $\beta\in \sol(\pi_{\tilde y}(R(\vec x, \vec y)))$ such that $\alpha|_{\tilde x} \cup \beta|_{\tilde y}\notin \sol(R(\vec x, \vec y))$. Simply extending $\alpha$ and $\beta$ to an assignment in $R(\vec x, \vec y)$ yields the claim.
\end{proof}

Since Claim~\ref{clm:noproduct} applies to all constraints in $F^*$, we are now in a situation where we can use Lemma~\ref{lem:lowerboundmatchingformula} which shows that any rectangle cover respecting the partition $(X_1, X_2)$ for $F^*$ has size $2^{|E^*|} = 2^{\Omega(|A|+|B|)}$. With Lemma~\ref{lem:DNNFmakeRectangles}, we know that $\|O'\| = 2^{\Omega(|A|+|B|)}$ and since the construction of $O'$ from $O$ does not increase the size of the DNNF, we get $s= \|O\| = 2^{\Omega(|A|+|B|)}=  2^{\Omega(\|F_{\varphi_n}\|)}$.
\end{proof}

\section{The Boolean Case}\label{sct:boolean}

In this section, we will specialize our dichotomy results for the Boolean domain $\{0,1\}$.
A relation $R$ over $\{0,1\}$ is called \emph{bijunctive affine} if it can be written as a conjunction of the relations $x=y$ and $x\ne y$ and unary relations, so $R \in \coclone{E}$ with $E = \{=, \neq, U_0, U_1\}$ where $U_0=\{0\}$ and $U_1= \{1\}$. A set of relations $\Gamma$ is called bijunctive affine if all $R\in \Gamma$ are bijunctive affine. We will show the following dichotomy for the Boolean case:

\begin{theorem}\label{thm:boolean}
    Let $\Gamma$ be a constraint language over the Boolean domain. If all relations in $\Gamma$ are bijunctive affine, then there is a polynomial time algorithm that, given
    a $\Gamma$-formula $\varphi$, constructs an \textup{OBDD} for $\varphi$. If not, then
    there is a family of $\Gamma$-formulas $\varphi_{n}$ and $\varepsilon>0$ such that a \textup{DNNF}
    for $F_{\varphi_{n}}$ needs to have a size of at least $2^{\varepsilon\left\Vert \varphi_{n}\right\Vert }$.
\end{theorem}
Let us remark here that, in contrast to general domains $D$, there is no advantage of FDD over ODD in the Boolean case: either a constraint language allows for efficient representation by ODD or it is hard even for DNNF. So in a sense, the situation over the Boolean domain is easier. Also note that the tractable cases over the Boolean domain are very restricted, allowing only equalities and disequalities.

\begin{proof}[Proof of Theorem~\ref{thm:boolean}]
    With Theorem~\ref{thm:ODD} and Theorem~\ref{thm:FDD}, we only need to show the following:
    \begin{enumerate}
        \item Every $\Gamma \subseteq \langle E\rangle$ is strongly uniformly blockwise decomposable.\label{it:boolean1}
        \item If $\Gamma$ is strongly blockwise decomposable, then $\Gamma \subseteq \langle E\rangle $.\label{it:boolean2}
    \end{enumerate}
    We show that $E$ is strongly uniformly blockwise decomposable --
    which implies that every $\Gamma \subseteq \langle E\rangle$ is
    also strongly uniformly blockwise decomposable.
    Let $\varphi(\vec u)=\pi_{\tilde u}\psi(\vec w)$ be a
    pp-definition over $E$.
    We have to show that $F_{\varphi}(\vec u)$ is uniformly blockwise
    decomposable and by Lemma~\ref{lem:closedProjectionSeletion}, it
    suffices to show this for the constraint $F_{\psi}(\vec w)$. To
    this end, consider the colored graph $G_{\psi}$ with vertex set
    $\tilde w$ and where two vertices $x$ and $y$ are connected with a
    blue edge if the representation of $\psi$ contains $x=y$ and connected with a red edge if $\psi$ contains $x\neq y$. A vertex $x$ is blue if $\psi$ contains $U_1(x)$ and red if $\psi$ contains $U_0(x)$. 
    \begin{example}\label{exa:boolean}
        If $\psi(a,b,c,d,e,f)$ is the conjunction 
        \[
        (a=b)\land (c=d)\land (b \neq c)\land (b \neq d)\land (e \neq
        f)\land  (b = 1) \land  (d = 0),
        \]
        then $G_\psi$ is the graph in Figure~\ref{fig:boolean}.
    \end{example}
    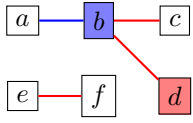
\begin{figure}
      \centering
        \begin{tikzpicture}
            \node[draw] (e) at (0,0) {$e$};
            \node[draw] (f) at (1,0) {$f$};
            \node[draw,fill=red!50] (d) at (2,0) {$d$};
            \node[draw] (a) at (0,1) {$a$};
            \node[draw,fill=blue!50] (b) at (1,1) {$b$};
            \node[draw] (c) at (2,1) {$c$};
            \begin{scope}[red, thick]
                \draw (e)--(f);
                \draw (c)--(b)--(d);
            \end{scope}
            \begin{scope}[blue, thick]
                \draw (a)--(b);
                \draw (c)--(d);
            \end{scope}
        \end{tikzpicture}
        \caption{The graph $G_{R'}$ from Example~\ref{exa:boolean}.}
        \label{fig:boolean}
    \end{figure}

    We show that $M_{x,y}^{F_\psi(\vec w)}$ is a proper block matrix and each block is decomposable using the same variable partition. If $x$ and $y$ are in different connected components of $G_{\psi}$, then $F_{\psi}(\vec w)$ is decomposable such that $x$ and $y$ appear in different factors, so $M_{x,y}^{F_\psi(\vec w)}$ has only one block which is decomposable. If $F_\psi(\vec w)$ has no satisfying assignments, there is nothing to show. It remains the case where $F_\psi(\vec w)$ has satisfying assignments and $x$ and $y$ are in the same connected component.
    
    So let $F_\psi(\vec w)$ have at least one model and let $G_{\psi}$ have only one connected component. Note that setting one variable in $\tilde w$ to $0$ or $1$ determines the values of all other variables in $\tilde w$. This implies that if a vertex in $G_\psi$ is colored, then $M_{x,y}^{F_\psi(\vec w)}$ has three empty entries and one entry with exactly one element. If no vertex in $G_{\psi}$ is colored, then $\pi_{x,y}(F_\psi(\vec w))$ is either $x=y$ (if every path from $x$ to $y$ has an even number of red edges) or $x\neq y$ (if every path from $x$ to $y$ has an odd number of edges). So $M_{x,y}^{F_\psi(\vec w)}$ has exactly two non-empty entries with one element each, so $R$ is decomposable in $x$ and $y$. This completes the proof of Item~\ref{it:boolean1}.
    
    Let now $\Gamma$ be strongly blockwise decomposable and $R\in \left< \Gamma \right>$ be decomposed into indecomposable factors
    \[
    R(\vec x) = R_1(\vec x_1) \times \cdots \times R_k(\vec x_k)
    \]
    We have to show that $R_i \in \langle E\rangle$ for all $i$ which implies that $R \in \langle E\rangle$. Since $R_i(\vec x_i)$ is indecomposable, $M_{x,y}^{R_i(\vec x_i)}$ has to have at least two blocks for every $x,y \in \tilde x_i$. Since $R$ is Boolean, only two cases remain: $\pi_{x,y} R_i(\vec x_i)$ is either equality or disequality on $x$ and $y$. In every case, the value of $y$ determines the value of $x$ and vice versa. Since $x$ and $y$ are arbitrary, each variable in $\tilde x_i$ determines the value of every other variable in $\tilde x_i$. This implies that non-empty entries of $M_{x,y}^{R_i}$ have exactly one element. It follows that $R_i(\vec x_i)$ has exactly two satisfying assignments and moreover, no variable can take the same value in these two assignments since otherwise $R_i(\vec x_i)$ would be decomposable. It follows that after fixing a value for a variable $x$, the values of all other variables $y$ are determined by the constraint $x\ne y$ or $x=y$, so $R_i(\vec x_i)$ is a conjunction of constraints with relations in $E$ which completes the proof of Item~\ref{it:boolean2} and thus the theorem.
\end{proof}

\section{Decidability of Strong (Uniform) Blockwise Decomposability}\label{sct:decidability}

In this section, we will show that strong (uniform) blockwise decomposability is decidable. 

\subsection{The Algorithm}

We will here first give the algorithm to decide strong (uniform) blockwise decomposability. The algorithm will rely on properties of constraint languages formulated below in Proposition~\ref{prop:ternary} which we will then prove in the following sections.
For our algorithm, we will heavily rely on the following well-known result for deciding containment in co-clones whose proof can be found in~\cite[Lemma~42]{Dalmau00}.

\begin{theorem}\label{thm:computableCocloneContainment}
    There is an algorithm that, given~a constraint language $\Gamma$ and~a relation $R$, decides if $R\in \coclone{\Gamma}$.
\end{theorem}

In a first step in our algorithm, we would like to restrict to the
case of binary relations by taking binary projections as it is done in
Proposition~\ref{prop:restrictionToBinary}. However, as the following example shows, this is in general not correct since there are constraint languages that are not blockwise decomposable while their binary projection is.

\begin{example}\label{exa:projectionArgh}
    Consider the relation $R^\oplus:= \{(0,0,1), (0,1,0), (1,0,0), (1,1,1)\}$ which is the parity relation on three variables. Then the constraint $R^\oplus(x,y,z)$ has the selection matrix
  \begin{align*}
    \label{eq:7}
    M_{x,y}^{R^\oplus(x,y,z)}=\left(\begin{array}{c|cc}
        x\backslash y & 0 & 1 \\
        \hline 0 & \left\{ 1\right\}  & \left\{ 0\right\}  \\
        1 & \left\{ 0\right\}  & \left\{ 1\right\}.
    \end{array}\right)
\end{align*}
The single block of this matrix is obviously not decomposable, so $R$ is not blockwise decomposable. 

Let us now compute $\Pi_{\le 2}(\{R^\oplus\})$. To this end, call a relation \emph{affine} if it is the solution set of a system of linear equations over the finite field $\mathbb{F}_2$ with $2$ elements. Obviously, all $\{R^\oplus\}$-formulas define affine relations. The following fact is well known, see e.g.~\cite[Lemma 5]{Dalmau00} for a proof.
\begin{fact*}
    Every projection of an affine constraint is affine.
\end{fact*}
It follows that $\Pi_{\le 2}(\{R^\oplus\})$ only contains affine relations of arity at most $2$. But, as we saw in the proof of Theorem~\ref{thm:boolean}, in that case $\Pi_{\le 2}(\{R^\oplus\})$ is strongly uniformly blockwise decomposable.
As a consequence, we have that $R^\oplus$ is \emph{not} blockwise decomposable while its binary projection is even strongly uniformly blockwise decomposable.
\end{example}

To avoid the problem of Example~\ref{exa:projectionArgh}, we will make
sure that for the constraint language at hand we have
$\coclone{\Gamma} = \coclone{\Pi_{\le 2}(\Gamma)}$, which we will test with the help of 
Theorem~\ref{thm:computableCocloneContainment} below. If this is
not the case, then by the contraposition of
Proposition~\ref{prop:restrictionToBinary}, we already know that the
language $\Gamma$ is not strongly blockwise decomposable.
Afterwards, we can focus on $\Pi_{\le 2}(\Gamma)$ and utilize the following
proposition, which is the main technical contribution of this section.

\newcommand{\Gammatwo}{\ensuremath{\Gamma\!_2}}

\begin{proposition}\label{prop:ternary}
    Let $\Gamma$ be a constraint language and
    $\Gammatwo:=\Pi_{\le 2}(\Gamma)$ the constraint language consisting of all unary and binary pp-definable relations over $\Gamma$.
    \begin{enumerate}
    \item $\Gammatwo$ is strongly uniformly blockwise decomposable if and
      only if all relations of arity at most $3$ in
      $\coclone{\Gammatwo}$ are uniformly blockwise decomposable.
    \item $\Gammatwo$ is strongly blockwise decomposable if and only if all relations of arity at most $3$ in $\coclone{\Gammatwo}$ are blockwise decomposable.
    \end{enumerate}
\end{proposition}

Before we prove Proposition~\ref{prop:ternary} in the next two subsections, let us show how it yields the desired decidability results.

\begin{theorem}\label{thm:decidability}
  There is an algorithm that, given a constraint language $\Gamma$, decides if $\Gamma$ is strongly blockwise decomposable. Moreover, there is also an algorithm that, given a constraint language $\Gamma$, decides if $\Gamma$ is strongly uniformly blockwise decomposable.
\end{theorem}
\begin{proof}
    We only consider blockwise decomposability since the proof for uniform decomposability is completely analogous.
    Let $\Gamma$ be the given constraint language over the domain $D$.
    \begin{enumerate}
    \item Compute $\Gammatwo:=\Pi_{\le 2}(\Gamma)$ by testing for every
      unary and binary relation $R$ over $D$ whether
      $R\in\coclone{\Gamma}$ using
      Theorem~\ref{thm:computableCocloneContainment}.
    \item Again using Theorem~\ref{thm:computableCocloneContainment},
      we test for every $R\in \Gamma$, whether
      $R\in\coclone{\Gammatwo}$. If the answer is no, then by
      we can conclude by Proposition~\ref{prop:restrictionToBinary} that
      $\Gamma$ is not strongly blockwise decomposable. Otherwise, we know that $\coclone{\Gamma}=\coclone{\Gammatwo}$.
    \item By applying Theorem~\ref{thm:computableCocloneContainment} a
      third time, we compute all at most ternary relations in
      $\coclone{\Gammatwo}$ and test whether they are blockwise
      decomposable by a brute-force application of Definition~\ref{def:blockwisedecomposable}. By
      Proposition~\ref{prop:ternary} this is the case if and only if
      $\Gammatwo$ and hence $\Gamma$ is strongly blockwise
      decomposable. \qedhere
    \end{enumerate}
\end{proof}

\subsection{Proof of Proposition~\ref{prop:ternary}.1 (the uniform case)}

In this section, we will show Proposition~\ref{prop:ternary} for the
case of strong uniform blockwise decomposability. Obviously, if
$\coclone{\Gammatwo}$ contains a relation of arity at most $3$ that is
not uniformly blockwise decomposable, then $\Gammatwo$ is not strongly
uniformly blockwise decomposable. So we only have to show the other
direction of the claim.

To that end, let us assume that all relations of arity at most 3 in $\coclone{\Gammatwo}$ are uniformly blockwise
decomposable and take $R \in \coclone{\Gammatwo}$. We want to show that R is strongly uniformly blockwise
decomposable. In view of Lemma~\ref{lem:closedProjectionSeletion}, it suffices to consider the case in which $R$ has a pp-definition without any projections, so there is a pp-definition of the constraint $R(\vec x)$ of the form
\begin{equation}
    \varphi(\vec x) := \bigwedge_{x,y\in \tilde x} R_{xy}(x,y),\label{eq:binaryDecomp}
\end{equation}
where $R_{xy}$ is a relation from $\Gammatwo$ (here we use that
$R^\triv:=D^2\in \Gammatwo$ as it is pp-definable and that $\Gammatwo$
is closed under intersections).
We show that in our setting we get a decomposition as in Lemma~\ref{lem:thereisatree}.

\begin{claim}\label{clm:treelike}
    For any $\hat x, \hat y \in \tilde x$ there is an undirected tree
    $T$ with vertex set $V(T) = \tilde x$ and edge $\{\hat x, \hat y\}\in E(T)$ such that $R(\vec x) = F_\psi(\vec x)$ with 
    \begin{align*}
        \psi(\vec x) =  \bigwedge_{\{p,q\} \in E(T)}
        \pi_{\{p,q\}}(R(\vec x)).
    \end{align*}
\end{claim}
\begin{proof}
    If $R$ has less than three variables, there is nothing to show, so
    we will first consider the case of $|\tilde x| = 3$, showing first the
    following slightly stronger statement \textbf{($\mathbf{\ast}$)}: 
    \begin{enumerate}
        \item[\textbf{($\mathbf{\ast}$)}] Let $R(x,y,z)$ be a ternary
          constraint with constraint relation in $\coclone{\Gammatwo}$, then in the representation (\ref{eq:binaryDecomp}) one of the constraints $R_{xz}$ or $R_{yz}$ is the trivial constraint $R^\triv$.
    \end{enumerate}
    Since we are proving the backward direction of Proposition~\ref{prop:ternary}.1, $R(x,y,z)$ is
    uniformly blockwise decomposable by assumption. Thus, we can apply
    Lemma \ref{lem:onlycheckprojections} to see that $\psi(x,y,z)$ can be rewritten in one of the forms
    \begin{align*}
        \psi(x,y,z) &= R'_{xz}(x,z) \land R'_{xy}(x,y) \text{ or}\\
        \psi(x,y,z) &= R'_{xy}(x,y) \land R'_{yz}(y,z).
    \end{align*}
    Assertion \textbf{($\mathbf{\ast}$)} immediately follows and
    setting $x=\hat x$ and $y=\hat y$ proves Claim~\ref{clm:treelike}
    for $|\tilde x| = 3$. 
 
    Now assume that $|\tilde x| \geq 4$. Choose any variable $z\in
    \tilde x\setminus\{\hat x, \hat y\}$ and let $\vec x'$ be the varible vector we get from $\vec x$ by deleting all occurrences of $z$. Consider the formula
    \begin{align*}
        \psi'(\vec x'):= \bigwedge_{x,y\in \tilde x'} R_{xy}(x,y).
    \end{align*}
    Then we have by definition that 
    \begin{align*}
        \psi(\vec x) = \psi'(\vec x') \land \bigwedge_{x\in \tilde x'} R_{xz}(x,z).
    \end{align*}
    We claim that we can rewrite $\psi(\vec x)$ such that only at most one of the $R_{xz}$ is not the trivial relation $R^\triv$. To see this, assume that there are two different variables $x,y\in \tilde x$ such that $R_{xz}$ and $R_{yz}$ are both nontrivial. Let $F_{\psi'}(\vec x')$ be the constraint defined by $\psi'(\vec x')$ and let $R_{xy}'(x,y) = \pi_{\{x,y\}}F_{\psi'}(\vec x')$. Consider the formula $\psi'':= R'_{xy}(x,y)\land R_{xz}(x,z) \land R_{yz}(y,z)$. Applying \textbf{($\mathbf{\ast}$)}, we get that we can rewrite $\psi''$ such that one of $R_{xz}$ or $R_{yz}$ is trivial. Substituting this rewrite in for $\psi''$ in $\psi$ and iterating the process yields that there is in the end only one non-trivial $R_{xz}$. Let $x^*$ be the only variable for which $R_{x^*z}$ might be non-trivial, then we can assume that 
    \begin{align*}
      F_\psi(\vec x) = F_\psi(\vec x') \land R_{x^* z}(x^*, z).
    \end{align*}
    Since $F'$ has fewer variables than $F$, we get by induction that
    there is a tree $T$ with vertex set $\tilde x\setminus \{z\}= \tilde x'$ and
    $\{\hat x, \hat y\}\in E(T)$ such that 
        \begin{align*}
        R'(\vec x') =  \bigwedge_{\{p,q\} \in E(T)}
        \pi_{\{p,q\}}(R'(\vec x')).
    \end{align*}
    Adding $z$ as a new leaf connected to $x^*$ gives the desired tree for $R(\vec x)$.
\end{proof}

Using Claim~\ref{clm:treelike}, we now show that $R(\vec x)$ is
uniformly blockwise decomposable. To this end, we fix two variables
$x,y\in \tilde x$ and show that $R(\vec x)$ is uniformly blockwise decomposable in $x,y$.
To see that $M^{R(\vec x)}_{x,y}$ is a proper block matrix, observe that $M^{R(\vec x)}_{x,y}$ has the same non-empty entries as $M^{\pi_{\{x,y\}}(R(\vec x))}_{x,y}$. But $\pi_{\{x,y\}}(R(\vec x))$ is in $\coclone{\Gammatwo}$ and thus by assumption uniformly blockwise decomposable. It follows that its selection matrix is a proper block matrix which is then also true for~$M^{R(\vec x)}_{x,y}$.

We now apply Claim~\ref{clm:treelike} to $x,y$ and let $T$ be the
resulting tree. Since $\{x,y\}$ is an edge in $T$,
$(V(T),E(T) - \{x,y\})$ consists of two trees $T_x$ and $T_y$,
containing $x$ and $y$, respectively. By setting
$U=V(T_x)\setminus\{x\}$ and $W=V(T_y)\setminus\{y\}$,
Claim~\ref{clm:treelike} implies that $R(\vec x)$ is defined by
\begin{align*}
        \theta(\vec x) =  \pi_{U\cup\{x\}}(R(\vec x)) \land
        \pi_{\{x,y\}}(R(\vec x))  \land \pi_{W\cup\{y\}}(R(\vec x)).
\end{align*}
This implies that each block of $M^{R(\vec x)}_{x,y}$ is decomposable in
$(U\cup\{x\},W\cup\{y\})$ and hence that $R(\vec x)$ is uniformly blockwise
decomposable.

\subsection{Proof of Proposition~\ref{prop:ternary}.2 (the nonuniform case)}

In this section, we prove Proposition~\ref{prop:ternary} for the case
of blockwise decomposability, so let $\Gamma$ be a
constraint language and $\Gammatwo = \Pi_{\le 2}(\Gamma)$. %
We will first define the following new property of $\Gammatwo$.

\begin{definition}[Incompatible block structure]\label{def:incompatibleblocks}
    We say that $\Gammatwo$ has an \emph{incompatible block structure} if and only if there are binary relations $R_1, R_2$ in $\Gammatwo$ such that $M^{R_1(x,z)}_{x,z}$ has blocks $(A_x, A_z)$ and $(B_x, B_z)$ and $M^{R_2(z,y)}_{z,y}$ has blocks $(C_z, C_y)$ and $(D_z, D_y)$ such that $A_z\cap C_z$, $A_z\cap D_z$, $B_z\cap C_z$, and $B_z\cap D_z$ are all non-empty.
\end{definition}

In the remainder of this section we will show that the following are
equivalent, the equivalence between (1) and (2) then establishes Proposition~\ref{prop:ternary}.2:

\begin{proposition}\label{prop:criterion}For every constraint language $\Gamma$, the following statements are equivalent:
\begin{itemize}
    \item[(1)] $\Gammatwo$ is strongly blockwise decomposable.
    \item[(2)] Every ternary relation in $\coclone{\Gammatwo}$ is blockwise decomposable.
    \item[(3)] $\Gammatwo$ is blockwise decomposable and has no incompatible block structure.
\end{itemize}
    
\end{proposition}
The direction (1)~$\Rightarrow$~(2) is trivial, (2)~$\Rightarrow$~(3)
will be shown in Lemma~\ref{lem:notPtoNotDecomposable}, and (3)~$\Rightarrow$~(1) is stated in Lemma~\ref{lem:PtoDecomposable}.

\begin{lemma}\label{lem:notPtoNotDecomposable}
    If $\Gammatwo$ has an incompatible block structure, then $\coclone{\Gammatwo}$ contains a ternary relation that is not blockwise decomposable.
\end{lemma}
\begin{proof}
    Let $R_1, R_2$ and the corresponding blocks be chosen as in Definition~\ref{def:incompatibleblocks}. We claim that the constraint 
    \begin{align*}
        R(x,y,z) := R_1(x,z)\land R_2(z,y)
    \end{align*}
    is not blockwise decomposable. To this end, choose values $a\in A_x, b\in B_x, c\in C_y, d\in D_y$. Then we have by construction that
    \begin{align*}
        \pi_{z}R(x,y,z)|_{x=a,y=c} & =A_{z}\cap C_{z}\\
        \pi_{z}R(x,y,z)|_{x=a,y=d} & =A_{z}\cap D_{z}\\
        \pi_{z}R(x,y,z)|_{x=b,y=c} & =B_{z}\cap C_{z}\\
        \pi_{z}R(x,y,z)|_{x=b,y=d} & =B_{z}\cap D_{z}
    \end{align*}
    and all these sets are all non-empty because of the incompatible block structure. Now assume, by way of contradiction, that $R(x,y,z)$ is blockwise decomposable. Then $M^{R(x,y z)}_{x,y}$ is a proper block matrix and the entries $(a,c), (a,d), (b,c), (b,d)$ all lie in the same block $B$. Then $B$ decoposes with respect to $x,y$, because we assumed that $R(x,y,z)$ is blockwise decomposable, so 
\begin{align*}
    R(x,y,z)|_{(x,y)\in B} & =\pi_{x,z}R(x,y,z)|_{(x,y)\in B}\times\pi_{y}R(x,y,z)|_{(x,y)\in B}\\
    \text{or}\quad R(x,y,z)|_{(x,y)\in B} & =\pi_{x}R(x,y,z)|_{(x,y)\in B}\times\pi_{z,y}R(x,y,z)|_{(x,y)\in B}
\end{align*}

However, in the first case we have \begin{align*}
    \pi_{z}R(x,y,z)|_{x=a,y=c}=\pi_{z}R(x,y,z)|_{x=a}=\pi_{z}R(x,y,z)|_{x=a,y=d}\ne \emptyset,
\end{align*} so in particular there is an element of $D$ in both $\pi_z R_2(z,y)|_{y=c} = C_z$ and $\pi_z R_2(z,y)|_{y=d} = D_z$ which contradicts the assumption that $(C_z, C_y)$ and $(D_z, D_y)$ are different blocks of $M^{R_2(z,y)}_{z,y}$.
In the second case, we get 
\begin{align*}
    \pi_{z}R(x,y,z)|_{x=a,y=c}=\pi_{z}R(x,y,z)|_{y=c}=\pi_{z}R(x,y,z)|_{x=b,y=c}\ne \emptyset
\end{align*}
which leads to an analogous contradiction to $(A_x, A_z)$ and $(B_x, B_z)$ being blocks of $R_1$. Thus, in both cases we get a contradiction, so $R(x,y,z)$ cannot be blockwise decomposable.
\end{proof}

\begin{lemma}\label{lem:PtoDecomposable}
    If $\Gammatwo$ is blockwise decomposable and has no incompatible block structure, then it is strongly blockwise decomposable.
\end{lemma}
\begin{proof}
    Let $R$ be in $\coclone{\Gammatwo}$. We will show that $R$
    is blockwise decomposable. Because of
    Lemma~\ref{lem:closedProjectionSeletion}, we may w.l.o.g.~assume
    that $R$ is pp-definable by a $\Gammatwo$-formula without
    projection. Thus, as in the proof of Proposition~\ref{prop:ternary}.1 we can write $R(\vec x)$ as $F_\phi(\vec x)$ with 
    \begin{align*}
        \varphi(\vec x) = \bigwedge_{x,y\in \tilde x} R_{xy}(x,y),
    \end{align*}
    where $R_{xy}\in \Gammatwo$. Now using again the fact that
    $\Gammatwo = \Pi_{\le 2}(\Gamma) = \Pi_{\le 2}(\Gammatwo)$ and thus in particular $\pi_{\{x,y\}} R(\vec x)\in \Gammatwo$ for all $x,y\in \tilde x$, we can actually write $R(\vec x)$ as $F_\psi(\vec x)$ for 
    \begin{align*}
    \psi(\vec x) = \bigwedge_{x,y\in \tilde x} \pi_{\{x,y\}} R(x,y).
    \end{align*}
Fix two variables $u,v$. Since $\Gammatwo$ is blockwise decomposable
by assumption of the lemma, we have that $M^{\pi_{\{u,v\}}(R(\vec
  x))}_{uv}$ and hence $M^{R(\vec x)}_{uv}$ is a proper block matrix (because
they have the same block structure). We have to show that all of its blocks decompose. So fix a block $B$ of $M^{R(\vec x)}_{uv}$ and consider the restriction $R^B(\vec x):= R(\vec x)|_{(u,v)\in B}$. As before, we get a representation
    \begin{align}
    R^B(\vec x) = \bigwedge_{x,y\in \tilde x} \pi_{\{x,y\}} R^B(x,y).\label{eq:propertyP}
\end{align}
We now assign a graph $G$ to $R^B(\vec x)$ as follows: vertices are the
variables in~$\tilde x$; two variables $x,y$ are connected by an edge
if and only if $M_{xy}^{R^B(\vec x)}$ has more than one block. Since $\Gammatwo$
is blockwise decomposable, we have, again using that $\Gamma_2 = \Pi_{\le 2}(\Gamma_2)$, that all $M^{R^B(\vec x)}_{xy}$ are proper
block matrices. So the variables $x,y$ that are not connected are
exactly those where $M^{R^B(\vec x)}_{xy}$ has exactly one block. So in
particular, there is no edge between $u$ and $v$. The crucial 
property is that the edge relation is transitive: 

\begin{claim}\label{claim:trans}
    For all $x,y,z\in \tilde x$, if $xz$ and $zy$ are edges in $G$, then $xy$ is also an edge.
\end{claim}
\begin{proof}
    Again, for all $x',y'\in \tilde x$, the matrices $M^{R^B(\vec x)}_{x'y'}$ and $M^{{\pi_{\{x',y'\}}R^B(\vec x)}}_{x'y'}$ have the same blocks. Thus, $x'y'$ is an edge in $G$ if and only if $M^{{\pi_{\{x',y'\}}R^B(\vec x)}}_{x'y'}$ has more than one block.
    
    Throughout the remainder of this proof, we will use the following observation which follows directly from the fact that we have $\pi_{\{z\}} R^B(\vec x)= \pi_{\{z\}} (\pi_{\{x,z\}}R^B(\vec x)) = \pi_{\{z\}}(\pi_{\{z,y\}}R^B(\vec x))$.
    \begin{observation}\label{obs:coverz}
        For every element $d\in \pi_{\{z\}}R^B(\vec x)$ there must be a block $(A_x, A_z)$ of $M^{\pi_{\{x,z\}}R^B(\vec x)}(x,z)$ and a block $(C_z, C_y)$ of $M^{\pi_{\{z, y\}}R^B(\vec x)}_{z,y}$ that both contain $d$.
    \end{observation}

Let $(A^1_x,A^1_z),\ldots,(A^p_x,A^p_z)$ be the blocks of
$M^{\pi_{\{x,z\}}R^B(\vec x)}_{x,z}$ and $(C^1_z,C^1_y),\ldots,(C^q_z,C^q_y)$ be the blocks of
$M^{\pi_{\{z,y\}}R^B(\vec x)}_{z,y}$. By assumption of the claim we have $p,q\geq 2$ and
by Observation~\ref{obs:coverz} $\bigcup^p_{\ell=1}A^\ell_z =
\bigcup^q_{\ell=1}C^\ell_z$. We claim that then there are two blocks
$(A^i_x,A^i_z)$ and $(A^j_x,A^j_z)$ such that for all $\ell \in [q]$ we have $A^i_z\cup A^j_z
\not\subseteq C^\ell_z$. To see this, consider two cases: if there is an $A^i_z$ that is not fully contained in any $C^\ell_z$, then we can choose an arbitrary other set $A^j_z$ to satisfy the claim. Otherwise, choose $i, \ell$ arbitrarily such that $A_z^i \subseteq C_z^\ell$. Since $q\ge 2$, there are elements from $D$ that are not in $C_z^\ell$ and thus there is an $A_z^j\nsubseteq C_z^\ell$. But then $A^i_z\cup A^j_z$ has the desired property since $A^i_z$ contains only elements from $C^\ell_z$. So in any case we can choose $i,j\in [p]$ such that for all $\ell\in [q]$ we have that $A^i_z\cup A^j_z
\not\subseteq C^\ell_z$.

It follows that we can fix two distinct
blocks $(C^s_z,C^s_y)$ and $(C^t_z,C^t_y)$ such that $A^i_z\cap
C^s_z\neq \emptyset$ and $A^j_z\cap C^t_z\neq \emptyset$. Since
$\Gammatwo$ does not have an incompatible block structure, it must be
the case that either $A^j_z\cap C^s_z =
\emptyset$ or  $A^i_z\cap C^t_z = \emptyset$. W.l.o.g. assume
the former. It follows that for any $a\in A^j_x$, $b\in C^s_y$ and
every domain element $c\in D$ it is not the case that $(a,c)\in
\pi_{\{x,z\}}R^B$ and $(c,b)\in
\pi_{\{z,y\}}R^B$. Therefore $(a,b)\notin \pi_{\{x,y\}}R^B$ and $M^{\pi_{\{x,y\}} R^B(\vec x)}_{xy}[a,b]=\emptyset$. Since by
the choice of $a$ and $b$ we have $a\in \pi_{\{x\}}R^B(\vec x)$ and $b\in
\pi_{\{y\}}R^B(\vec x)$, we get that $M^{\pi_{\{x,y\}} R^B(\vec x)}_{xy}$ has
non-empty entries in row $a$ and column $b$ but not at their intersection, implying that the matrix contains at least
two blocks. This finishes the proof of Claim~\ref{claim:trans}.
\end{proof}
It remains to prove that $R^B(\vec x)$ is decomposable w.r.t.~some partition $(U,V)$
with $u\in U$ and $v\in V$. To this end, we let $U$ be the connected
component of $u$ in the graph $G$ and $V:=\tilde x\setminus U$. Note
that $v\in V$ because $uv$ is not an edge and the edge relation is
transitive. For every $x\in U$ and $y\in V$ the matrix
$M^{\pi_{\{x,y\}} R^B(\vec x)}_{xy}$ has exactly one block and therefore
$\pi_{\{x,y\}} R^B(\vec x) = \pi_{\{x\}} R^B(\vec x) \times \pi_{\{y\}} R^B(\vec x)$. This
implies that 
    \begin{align*}
    R^B(\vec x) &= \bigwedge_{x,y\in \tilde x} \pi_{\{x,y\}} R^B(\vec x)\\
                &= \bigwedge_{x,y\in U} \pi_{\{x,y\}} R^B(\vec x)\land
                  \bigwedge_{x,y\in V} \pi_{\{x,y\}} R^B(\vec x),
     \end{align*}
proving that $R^B(\vec x)$ is decomposable w.r.t.~$(U,V)$. 
\end{proof}

\subsection{A separating example}

We have established that strong blockwise decomposability as well
as strong uniform blockwise decomposability are decidable. It follows
that there is an algorithm that decides for a given constraint language $\Gamma$
if either
\begin{itemize}
\item 
every CSP($\Gamma$) instance can be encoded into a polynomial-size ODD or
\item
every CSP($\Gamma$) instance can be encoded into a polynomial-size FDD,
but some
CSP($\Gamma$) instances require exponential-size ODDs (and structured
DNNFs) or
\item
there are CSP($\Gamma$) instances that require exponential-size FDDs
(and DNNFs).
\end{itemize}
We have seen that there are constraint languages falling in the first
and third category. Furthermore, there is no \emph{Boolean} constraint
language falling in the second category. However, in
the non-Boolean case there are constraint languages with this property. 
To see this, we utilize our new criterion for strong blockwise
decomposability stated in Proposition~\ref{prop:criterion}: A constraint language is strongly blockwise
decomposable if and only if every binary relation in $\Pi_{\le 2}(\Gamma)$ is
blockwise decomposable (which is the same as being \emph{rectangular} \cite{DyerR13}) and there are no two relations in $\Pi_{\le 2}(\Gamma)$
with an incompatible block structure. Let us come back to our running
example $\Gamma=\{R\}$ from Section~\ref{sec:basic-properties}, see
Figure~\ref{fig:bin-projections}. We have already observed in Example~\ref{exa:separateNotions} that
$R$ is not uniformly blockwise decomposable and hence $\Gamma$ is not
strongly uniformly blockwise decomposable. Moreover, we have computed $\Pi_{\le 2}(\Gamma)=\{D,R_{\text{triv}}, R_=, R', R'', R'''\}$ in
Example~\ref{exa:projections2}. By inspecting these relations (see Figure~\ref{fig:no-incompatible-blocks}) it follows that they are
all blockwise decomposable and that no two relations have pairs of incompatible
blocks as stated in Definition~\ref{def:incompatibleblocks}. It
follows by Lemma~\ref{lem:PtoDecomposable} that $\Gamma$ is strongly
blockwise decomposable and thus $\Gamma$ serves as a separating example
of our two central notions. This leads to the following theorem, which
should be contrasted with the Boolean case where both notions collapse
(Theorem~\ref{thm:boolean}).

    \begin{figure}
      \centering
      \begin{tikzpicture}

        \node[minimum size = 4mm, inner sep=0pt, circle] (1d) at (1,0*0.5) {$d$};
        \node[minimum size = 4mm, inner sep=0pt, circle] (1c) at (1,1*0.5) {$c$};
        \node[minimum size = 4mm, inner sep=0pt, circle] (1b) at (1,2*0.5) {$b$};
        \node[minimum size = 4mm, inner sep=0pt, circle] (1a) at (1,3*0.5) {$a$};

        \node[minimum size = 4mm, inner sep=0pt, circle] (2d) at (2,0*0.5) {$d$};
        \node[minimum size = 4mm, inner sep=0pt, circle] (2c) at (2,1*0.5) {$c$};
        \node[minimum size = 4mm, inner sep=0pt, circle] (2b) at (2,2*0.5) {$b$};
        \node[minimum size = 4mm, inner sep=0pt, circle] (2a) at (2,3*0.5) {$a$};

        \draw[-,semithick] (1a) -- (2a);
        \draw[-,semithick] (1b) -- (2b);
        \draw[-,semithick] (1c) -- (2c);
        \draw[-,semithick] (1d) -- (2d);

        \draw[-,semithick] (1c) -- (2d);
        \draw[-,semithick] (1d) -- (2c);

        \node[minimum size = 4mm,inner sep=0pt, circle] (3d) at (3,0*0.5) {$d$};
        \node[minimum size = 4mm,inner sep=0pt, circle] (3c) at (3,1*0.5) {$c$};
        \node[minimum size = 4mm,inner sep=0pt, circle] (3b) at (3,2*0.5) {$b$};
        \node[minimum size = 4mm,inner sep=0pt, circle] (3a) at (3,3*0.5) {$a$};

        \node[minimum size = 4mm,inner sep=0pt, circle] (4d) at (4,0*0.5) {$d$};
        \node[minimum size = 4mm,inner sep=0pt, circle] (4c) at (4,1*0.5) {$c$};
        \node[minimum size = 4mm,inner sep=0pt, circle] (4b) at (4,2*0.5) {$b$};
        \node[minimum size = 4mm,inner sep=0pt, circle] (4a) at (4,3*0.5) {$a$};

        \draw[-,semithick] (3a) -- (4a);
        \draw[-,semithick] (3b) -- (4b);
        \draw[-,semithick] (3c) -- (4c);
        \draw[-,semithick] (3d) -- (4d);

        \draw[-,semithick] (3a) -- (4b);
        \draw[-,semithick] (3b) -- (4a);

        \node[minimum size = 4mm,inner sep=0pt, circle] (5d) at (5,0*0.5) {$d$};
        \node[minimum size = 4mm,inner sep=0pt, circle] (5c) at (5,1*0.5) {$c$};
        \node[minimum size = 4mm,inner sep=0pt, circle] (5b) at (5,2*0.5) {$b$};
        \node[minimum size = 4mm,inner sep=0pt, circle] (5a) at (5,3*0.5) {$a$};

        \node[minimum size = 4mm,inner sep=0pt, circle] (6d) at (6,0*0.5) {$d$};
        \node[minimum size = 4mm,inner sep=0pt, circle] (6c) at (6,1*0.5) {$c$};
        \node[minimum size = 4mm,inner sep=0pt, circle] (6b) at (6,2*0.5) {$b$};
        \node[minimum size = 4mm,inner sep=0pt, circle] (6a) at (6,3*0.5) {$a$};

        \draw[-,semithick] (5a) -- (6a);
        \draw[-,semithick] (5b) -- (6b);
        \draw[-,semithick] (5c) -- (6c);
        \draw[-,semithick] (5d) -- (6d);

        \draw[-,semithick] (5a) -- (6b);
        \draw[-,semithick] (5b) -- (6a);
        \draw[-,semithick] (5c) -- (6d);
        \draw[-,semithick] (5d) -- (6c);

        \node[minimum size = 4mm,inner sep=0pt, circle] (7d) at (7,0*0.5) {$d$};
        \node[minimum size = 4mm,inner sep=0pt, circle] (7c) at (7,1*0.5) {$c$};
        \node[minimum size = 4mm,inner sep=0pt, circle] (7b) at (7,2*0.5) {$b$};
        \node[minimum size = 4mm,inner sep=0pt, circle] (7a) at (7,3*0.5) {$a$};

        \node[minimum size = 4mm,inner sep=0pt, circle] (8d) at (8,0*0.5) {$d$};
        \node[minimum size = 4mm,inner sep=0pt, circle] (8c) at (8,1*0.5) {$c$};
        \node[minimum size = 4mm,inner sep=0pt, circle] (8b) at (8,2*0.5) {$b$};
        \node[minimum size = 4mm,inner sep=0pt, circle] (8a) at (8,3*0.5) {$a$};

        \draw[-,semithick] (7a) -- (8a);
        \draw[-,semithick] (7b) -- (8b);
        \draw[-,semithick] (7c) -- (8c);
        \draw[-,semithick] (7d) -- (8d);

        \node[minimum size = 4mm,inner sep=0pt, circle] (9d) at (9,0*0.5) {$d$};
        \node[minimum size = 4mm,inner sep=0pt, circle] (9c) at (9,1*0.5) {$c$};
        \node[minimum size = 4mm,inner sep=0pt, circle] (9b) at (9,2*0.5) {$b$};
        \node[minimum size = 4mm,inner sep=0pt, circle] (9a) at (9,3*0.5) {$a$};

        \node[minimum size = 4mm,inner sep=0pt, circle] (10d) at (10,0*0.5) {$d$};
        \node[minimum size = 4mm,inner sep=0pt, circle] (10c) at (10,1*0.5) {$c$};
        \node[minimum size = 4mm,inner sep=0pt, circle] (10b) at (10,2*0.5) {$b$};
        \node[minimum size = 4mm,inner sep=0pt, circle] (10a) at (10,3*0.5) {$a$};

        \draw[-,semithick] (9a) -- (10a);
        \draw[-,semithick] (9b) -- (10a);
        \draw[-,semithick] (9c) -- (10a);
        \draw[-,semithick] (9d) -- (10a);

        \draw[-,semithick] (9a) -- (10b);
        \draw[-,semithick] (9b) -- (10b);
        \draw[-,semithick] (9c) -- (10b);
        \draw[-,semithick] (9d) -- (10b);

        \draw[-,semithick] (9a) -- (10c);
        \draw[-,semithick] (9b) -- (10c);
        \draw[-,semithick] (9c) -- (10c);
        \draw[-,semithick] (9d) -- (10c);

        \draw[-,semithick] (9a) -- (10d);
        \draw[-,semithick] (9b) -- (10d);
        \draw[-,semithick] (9c) -- (10d);
        \draw[-,semithick] (9d) -- (10d);
        
        \node at (1.5,4*0.5) {$R'$};
        \node at (3.5,4*0.5) {$R''$};
        \node at (5.5,4*0.5) {$R'''$};
        \node at (7.5,4*0.5) {$R_=$};
        \node at (9.5,4*0.5) {$R_{\text{triv}}$};

        \draw[rectangle, rounded corners] (1a.north west)+(-.2mm,.5mm) rectangle (1a.south east)+(.3mm,-.5mm);
        \draw[rectangle, rounded corners] (1b.north west)+(-.2mm,.5mm) rectangle (1b.south east)+(.3mm,-.5mm);
        \draw[rectangle, rounded corners] (1c.north west)+(-.2mm,.5mm) rectangle (1d.south east)+(.3mm,-.5mm);
        \draw[rectangle, rounded corners] (2a.north west)+(-.2mm,.5mm) rectangle (2a.south east)+(.3mm,-.5mm);
        \draw[rectangle, rounded corners] (2b.north west)+(-.2mm,.5mm) rectangle (2b.south east)+(.3mm,-.5mm);
        \draw[rectangle, rounded corners] (2c.north west)+(-.2mm,.5mm) rectangle (2d.south east)+(.3mm,-.5mm);

        \draw[rectangle, rounded corners] (3a.north west)+(-.2mm,.5mm) rectangle (3b.south east)+(.3mm,-.5mm);
        \draw[rectangle, rounded corners] (3c.north west)+(-.2mm,.5mm) rectangle (3c.south east)+(.3mm,-.5mm);
        \draw[rectangle, rounded corners] (3d.north west)+(-.2mm,.5mm) rectangle (3d.south east)+(.3mm,-.5mm);
        \draw[rectangle, rounded corners] (4a.north west)+(-.2mm,.5mm) rectangle (4b.south east)+(.3mm,-.5mm);
        \draw[rectangle, rounded corners] (4c.north west)+(-.2mm,.5mm) rectangle (4c.south east)+(.3mm,-.5mm);
        \draw[rectangle, rounded corners] (4d.north west)+(-.2mm,.5mm) rectangle (4d.south east)+(.3mm,-.5mm);

        \draw[rectangle, rounded corners] (5a.north west)+(-.2mm,.5mm) rectangle (5b.south east)+(.3mm,-.5mm);
        \draw[rectangle, rounded corners] (5c.north west)+(-.2mm,.5mm) rectangle (5d.south east)+(.3mm,-.5mm);
        \draw[rectangle, rounded corners] (6a.north west)+(-.2mm,.5mm) rectangle (6b.south east)+(.3mm,-.5mm);
        \draw[rectangle, rounded corners] (6c.north west)+(-.2mm,.5mm) rectangle (6d.south east)+(.3mm,-.5mm);

        \draw[rectangle, rounded corners] (7a.north west)+(-.2mm,.5mm) rectangle (7a.south east)+(.3mm,-.5mm);        
        \draw[rectangle, rounded corners] (7b.north west)+(-.2mm,.5mm) rectangle (7b.south east)+(.3mm,-.5mm);        
        \draw[rectangle, rounded corners] (7c.north west)+(-.2mm,.5mm) rectangle (7c.south east)+(.3mm,-.5mm);        
        \draw[rectangle, rounded corners] (7d.north west)+(-.2mm,.5mm) rectangle (7d.south east)+(.3mm,-.5mm);        
        \draw[rectangle, rounded corners] (8a.north west)+(-.2mm,.5mm) rectangle (8a.south east)+(.3mm,-.5mm);        
        \draw[rectangle, rounded corners] (8b.north west)+(-.2mm,.5mm) rectangle (8b.south east)+(.3mm,-.5mm);        
        \draw[rectangle, rounded corners] (8c.north west)+(-.2mm,.5mm) rectangle (8c.south east)+(.3mm,-.5mm);        
        \draw[rectangle, rounded corners] (8d.north west)+(-.2mm,.5mm) rectangle (8d.south east)+(.3mm,-.5mm);        
        
        \draw[rectangle, rounded corners] (9a.north west)+(-.2mm,.5mm) rectangle (9d.south east)+(.3mm,-.5mm);                
        \draw[rectangle, rounded corners] (10a.north west)+(-.2mm,.5mm) rectangle (10d.south east)+(.3mm,-.5mm);                

      \end{tikzpicture}
      \caption{All binary relations in $\Pi_{\le 2}(R)$ and their compatible
        block structure.
        (Example~\ref{exa:projections2})}
      \label{fig:no-incompatible-blocks}
    \end{figure}
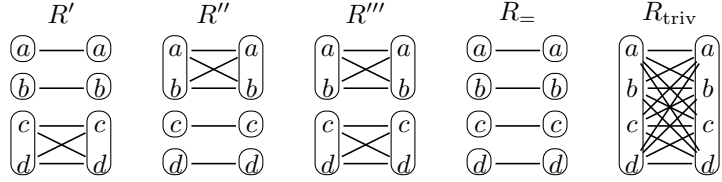

\begin{theorem}\label{thm:separation}
  There is a constraint language $\Gamma$ over a 4-element domain that is strongly
  blockwise decomposable, but not strongly uniformly
  blockwise decomposable. Thus, every \textup{CSP($\Gamma$)} instance can be
  decided by a polynomial-size \textup{FDD}, but there are \textup{CSP($\Gamma$)}
  instances that require structured \textup{DNNF}s (and \textup{ODD}s) of exponential size.
\end{theorem}

\section{Conclusion}\label{sct:conclusion}

We have seen that there is a dichotomy for compiling systems of constraints into DNNF based on the constraint languages. It turns out that the constraint languages that allow efficient compilation are rather restrictive, in the Boolean setting they consist essentially only of equality and disequality. From a practical perspective, our results are thus largely negative since interesting settings will most likely lie outside the tractable cases we have identified. 

Within the polynomially compilable constraint languages we have identified and
separated two categories, depending on whether they guarantee polynomial-size
structured representations. Moreover, both properties are decidable.

A few questions remain open. The first is to get a better grasp on the
efficiently compilable constraint languages. Is there a simpler
combinatorial description, or is there an algebraic characterization
using polymorphisms? Is there a simpler way of testing strong
(uniform) blockwise decomposability that avoids
Theorem~\ref{thm:computableCocloneContainment}? What is the exact
complexity?

\bibliographystyle{plain}
\bibliography{Paper}
\end{document}